\documentclass[runningheads,a4paper]{llncs}

\usepackage{amsfonts,amsmath,amssymb}
\usepackage{algorithmicx}
\usepackage{algorithm}
\usepackage{algpseudocode}
\usepackage{verbatim}
\usepackage{mathtools}
\usepackage{graphicx}
\usepackage{breqn}
\usepackage{subcaption}
\usepackage{enumitem}

\usepackage{url}
\urldef{\mailsa}\path|{fan, osimpson}@ucsd.edu|
\newcommand{\keywords}[1]{\par\addvspace\baselineskip
\noindent\keywordname\enspace\ignorespaces#1}

\def\polylog{\operatorname{polylog}}
\newcommand{\R}{\mathbb{R}}
\newcommand{\E}{\mathbb{E}}
\renewcommand{\P}{\mathbb{P}}

\newcommand{\vol}{\textmd{vol}}
\newcommand{\supp}{\textmd{supp}}
\renewcommand{\L}{\mathcal{L}}
\newcommand{\pr}{\textsf{pr}}
\newcommand{\localcheeg}{\Phi^*(S)}

\newcommand{\hkprseedalg}{\texttt{ApproxHKPRseed}}
\newcommand{\hkprseedalgparams}{\texttt{ApproxHKPRseed($G,t,u,\epsilon$)}}
\newcommand{\rparam}{\frac{16}{\epsilon^3}\log n}
\newcommand{\kparam}{\frac{\log(\epsilon^{-1})}{\log\log(\epsilon^{-1})}}
\newcommand{\hkprcomplexity}{O\big(\frac{\log(\epsilon^{-1})\log n}{\epsilon^3\log\log(\epsilon^{-1})}\big)}
\newcommand{\partitionalg}{\texttt{ClusterHKPR}}
\newcommand{\partitionalgparams}{\texttt{ClusterHKPR($G,u,s,\varsigma,\phi,\epsilon$)}}
\newcommand{\tparamcluster}{\phi^{-1}\log(\frac{2\sqrt{\varsigma}}{1-\epsilon} + 2\epsilon s)}

\begin{document}
\mainmatter

\title{Computing Heat Kernel Pagerank and a Local Clustering
Algorithm\footnote{An extended abstract appeared in~\cite{chung2014computing}.}}
\titlerunning{Computing Heat Kernel Pagerank}

\author{Fan Chung
\and Olivia Simpson}
\authorrunning{Chung and Simpson}

\institute{University of California, San Diego\\
La Jolla, CA 92093\\
\mailsa}

\toctitle{Computing Heat Kernel Pagerank}
\tocauthor{Chung and Simpson}
\maketitle

\begin{abstract}
Heat kernel pagerank is a variation of Personalized PageRank given in an
exponential formulation.  In this work, we present a sublinear time algorithm
for approximating the heat kernel pagerank of a graph.  The algorithm works by
simulating random walks of bounded length and runs in time $\hkprcomplexity$,
assuming performing a random walk step and sampling from a distribution with
bounded support take constant time.

The quantitative ranking of vertices obtained with heat kernel pagerank can be
used for local clustering algorithms.  We present an efficient local clustering
algorithm that finds cuts by performing a sweep over a heat kernel pagerank
vector, using the heat kernel pagerank approximation algorithm as a subroutine.
Specifically, we show that for a subset $S$ of Cheeger ratio $\phi$, many
vertices in $S$ may serve as seeds for a heat kernel pagerank vector which will
find a cut of conductance $O(\sqrt{\phi})$.
\end{abstract}

\section{Introduction}
\label{sec:introduction}
In large networks, many similar elements can be identified to a single, larger
entity by the process of clustering.  Increasing granularity in massive networks
through clustering eases operations on the network.  There is a large literature
on the problem of identifying clusters in a graph
(\cite{csz:spectralkway:94,shi2000normalized,njw:spectralcluster:02,kvv:clusterings:04,lc:powercluster:10,lc:textcluster:10}),
and the problem has found many applications.  However, in a variation of the
graph clustering problem we may only be interested in a single cluster near one
element in the graph.  For this, local clustering algorithms are of greater use.

As an example, the problem of finding a local cluster arises in protein
networks.  A protein-protein interaction (PPI) network has undirected edges that
represent an interaction between two proteins.  Given two PPI networks, the goal
of the pairwise alignment problem is to identify an optimal mapping between the
networks that best represents a conserved biological function.
In~\cite{liao:protein:09}, a local clustering algorithm is applied from a
specified protein to identify a group similar to that protein.  Such local
alignments are useful for analysis of a particular component of a biological
system (rather than at a systems level which will call for a global alignment).
Local clustering is also a common tool for identifying communities in a network.
A community is loosely defined as a subset of vertices in a graph which are more
strongly connected internally than to vertices outside the subset.  Properties
of community structure in large, real world networks have been studied
in~\cite{lldm:localnetwork:08}, for example, where local clustering algorithms
are employed for identifying communities of varying quality.

The goal of a local clustering algorithm is to identify a cluster in a graph
near a specified vertex.  Using only local structure avoids unnecessary
computation over the entire graph.  An important consequence of this are running
times which are often in terms of the size of the small side of the partition,
rather than of the entire graph.  The best performing local clustering
algorithms use probability diffusion processes over the graph to determine
clusters (see Section~\ref{sec:previouswork}).  In this paper we present a new
algorithm which identifies a cut near a specified vertex with simple
computations over a heat kernel pagerank vector.

The theory behind using heat kernel pagerank for computing local clusters has
been considered in previous work.  Here we give an efficient approximation
algorithm for computing heat kernel pagerank.  Note that we use a ``relaxed''
notion of approximation which allows us to derive a sublinear probabilistic
approximation algorithm for heat kernel pagerank, while computing an exact or
sharp approximation would require computation complexity of order similar to
matrix multiplication.  We use this sublinear approximation algorithm for
efficient local clustering.

\subsection{Previous work}\label{sec:previouswork}
\paragraph{Heat kernel and approximation of matrix exponentials.} Heat kernel
pagerank was first introduced in~\cite{chung:hkpr:07} as a variant of
personalized PageRank~\cite{haveliwala2002topic}.  While PageRank can be
viewed as a geometric sum of random walks, the heat kernel pagerank is an
exponential sum of random walks.  An alternative interpretation of the heat
kernel pagerank is related to the heat kernel of a graph as the fundamental
solution to the heat equation.  As such, it has connections with diffusion and
mixing properties of graphs and has been incorporated into a number of graph
algorithmic primitives.

Orecchia et al. use a variant of heat kernel random walks in their randomized
algorithm for computing a cut in a graph with prescribed balance
constraints~\cite{osv:balsep:11}.  A key subroutine in the algorithm is a
procedure for computing $e^{-A}v$ for a positive semidefinite matrix $A$ and a
unit vector $v$ in time $\tilde{O}(m)$ for graphs on $n$ vertices and $m$ edges.
They show how this can be done with a small number of computations of the form
$A^{-1}v$ and applying the Spielman-Teng linear
solver~\cite{st:graphpartitioning:stoc04}.  Their main result is a randomized
algorithm that outputs a balanced cut in time $O(m\polylog n)$.  In a follow up
paper, Sachdeva and Vishnoi~\cite{sachdeva2013matrix} reduce inversion of
positive semidefinite matrices to matrix exponentiation, thus proving that
matrix exponentiation and matrix inversion are equivalent to polylog factors.
In particular, the nearly-linear running time of the balanced separator
algorithm depends upon the nearly-linear time Spielman-Teng solver.

Another method for approximating matrix exponentials is given by Kloster and
Gleich in~\cite{kloster:columnexp:waw13}.  They use a Gauss-Southwell iteration
to approximate the Taylor series expansion of the column vector $e^{P}e_c$ for
transition probability matrix $P$ and $e_c$ a standard basis vector.  The
algorithm runs in sublinear time assuming the maximum degree of the network is
$O(\log\log n)$.

\paragraph{Local clustering.} Local clustering algorithms were introduced
in~\cite{st:graphpartitioning:stoc04}, where Spielman and Teng present a
nearly-linear time algorithm for finding local partitions with certain balance
constraints.  Let $\Phi(S)$ denote the cut ratio of a subset $S$ that we will
later define as the Cheeger ratio.  Then, given a graph and a subset of vertices
$S$ such that $\Phi(S) < \phi$ and $\vol(S)\leq \vol(G)/2$, their algorithm
finds a set of vertices $T$ such that $\vol(T) \geq \vol(S)/2$ and $\Phi(T) \leq
O(\phi^{1/3}\log^{O(1)}n)$ in time $O(m(\log n/\phi)^{O(1)})$.  This seminal
work incorporates the ideas of Lov\'asz and
Simonovitz~\cite{ls:mixingisoperimetric:90,ls:randomwalks:93} on isoperimetric
properties of random walks, and their algorithm works by simulating truncated
random walks on the graph.  Spielman and Teng later improve their approximation
guarantee to $O(\phi^{1/2}\log^{3/2}n)$ in a revised version of the
paper~\cite{st:localcluster:08}.

The algorithm of~\cite{st:graphpartitioning:stoc04,st:localcluster:08} improves
the spectral methods of~\cite{donath1972algorithms} and a similar expression
in~\cite{aloniso85} which use an eigenvector of the graph Laplacian to partition
the vertices of a graph.  However, the local approach of Spielman and Teng
allows us to identify focused clusters without investigating the entire graph.
For this reason, the running time of this and similar local algorithms are
proportional to the size of the small side of the cut, rather than the entire
graph.

Andersen et al.~\cite{acl:prgraphpartition:focs06} give an improved local
clustering algorithm using approximate PageRank vectors.  For a vertex subset
$S$ with Cheeger ratio $\phi$ and volume $k$, they show that a PageRank vector
can be used to find a set with Cheeger ratio $O(\phi^{1/2}\log^{1/2}k)$.  Their
local clustering algorithm runs in time $O(\phi^{-1}m\log^4 m)$.  The analysis
of the above process was strengthened in~\cite{ac:sharpdrops:07} and emphasized
that vertices with higher PageRank values will be on the same side of the cut as
the starting vertex.

Andersen and Peres~\cite{ap:evolving:09} later simulate a volume-biased evolving
set process to find sparse cuts.  Although their approximation guarantee is the
same as that of~\cite{acl:prgraphpartition:focs06}, their process yields a
better ratio between the computational complexity of the algorithm on a given
run and the volume of the output set.  They call this value the
\emph{work/volume ratio}, and their evolving set algorithm achieves an expected
ratio of $O(\phi^{-1/2}\log^{3/2}n)$.  This result is improved by Gharan and
Trevisan in~\cite{gt:optimalcluster:12} with an algorithm that finds a set of
conductance at most $O(\epsilon^{-1/2}\phi^{1/2})$ and achieves a work/volume
ratio of $O(\varsigma^{\epsilon}\phi^{-1/2}\log^2n)$ for target volume
$\varsigma$ and target conductance $\phi$.  The complexity of their algorithm is
achieved by running copies of an evolving set process in parallel.

\subsection{Our contributions}
In this paper, we give a probabilistic approximation algorithm for computing a
vector that yields a ranking of vertices close to the heat kernel pagerank
vector.  The approximation algorithm, \hkprseedalg, works by simulating random
walks and computing contributions of these walks for each vertex in the graph.
Assuming access to a constant-time query which returns the destination of a heat
kernel random walk starting from a specified vertex, \hkprseedalg~runs in time
$\hkprcomplexity$. In the context of this paper, we strictly address heat kernel
pagerank with a single vertex as a seed -- an analogy to Personalized PageRank
with total preference given to a single vertex.  Note that heat kernel pagerank
with a general preference vector (see Section~\ref{sec:preliminaries}) is a
combination of heat kernel pagerank with a single seed vertex.  We refer the
reader to~\cite{cs:imlinear:14} for this more general case.

Using \hkprseedalg~as a subroutine, we then present a local clustering algorithm
that uses a ranking according to an approximate heat kernel pagerank.  Let $G$
be a graph and $S$ a proper vertex subset with volume $\varsigma \leq \vol(G)/4$
and Cheeger ratio $\Phi(S) \leq \phi$.  Then, with probability at least
$1-\epsilon$, our algorithm outputs either a cutset $T$ with $\vol(T) \geq
\vol(S)/2$ and $\varsigma$-local Cheeger ratio at most $O(\sqrt{\phi})$ or a
certificate that no such set exists.  The algorithm has work/volume ratio of
$O(\varsigma^{-1}\epsilon^{-3}\log
n\log(\epsilon^{-1})\log\log(\epsilon^{-1}))$.  This result is formalized in
Theorem~\ref{thm:localpart}.  A summary of previous results and our
contributions are given in Table~\ref{table:resultssummary}.

\begin{table}
\centering
\begin{tabular}{c|c|c}\hline
Algorithm & Conductance of output set & Work/volume ratio\\\hline
\cite{st:localcluster:08} & $O(\phi^{1/2}\log^{3/2}n)$ & $O(\phi^{-2}\polylog n)$\\
\cite{acl:prgraphpartition:focs06} & $O(\phi^{1/2}\log^{1/2}n)$ & $O(\phi^{-1}\polylog n)$\\
\cite{ap:evolving:09} & $O(\phi^{1/2}\log^{1/2} n)$ & $O(\phi^{-1/2}\polylog n)$\\
\cite{gt:optimalcluster:12} & $O(\epsilon^{-1/2}\phi^{1/2})$ & $O(\varsigma^{\epsilon}\phi^{-1/2}\polylog n)$\\
This work & $O(\phi^{1/2})$ & $O(\varsigma^{-1}\epsilon^{-3}\log n\log(\epsilon^{-1})\log\log(\epsilon^{-1}))$\\\hline
\end{tabular}
\caption{Summary of local clustering algorithms}
\label{table:resultssummary}
\end{table}

As a summary of the contributions of this work,
\begin{enumerate}[label={(\arabic*)}]
\item We present an algorithm for computing a heat kernel pagerank vector from a
single seed vertex with $(1+\epsilon)$ approximation guarantee with high
probability in time $\hkprcomplexity$.\label{point:hkpralg} 
\item We present a local clustering algorithm which uses a ranking according
to heat kernel pagerank.  In our clustering algorithm we use the probabilistic
approximation algorithm in \ref{point:hkpralg} as a subroutine, which gives a
sublinear-time local clustering algorithm.\label{point:clusteralg} 
\item Using the approximation guarantees of \ref{point:hkpralg} and the analysis
for \ref{point:clusteralg}, we present a local clustering algorithm which with
high probability returns a set with Cheeger ratio at most $O(\sqrt{\phi})$,
given a target ratio $\phi$, with work/volume ratio
$O(\varsigma^{-1}\epsilon^{-3}\log n\log(\epsilon^{-1})\log\log(\epsilon^{-1}))$
where $\varsigma$ is proportional to the volume of the output set.
\item We validate the performance analysis by implementing our algorithms using
several real and synthetic graphs as examples.  The clusters that were derived
in these examples using the local clustering algorithm and heat kernel pagerank
approximation have Cheeger ratios as guaranteed in Theorem~\ref{thm:localpart}.
\end{enumerate}

The theory behind finding local cuts with heat kernel pagerank vectors was first
presented in~\cite{chung:hkpr:07,chung:partitionhkpr:im09}.  Using some of this
analysis as a starting point, we provide the algorithm for computing local
clusters, called \partitionalg.

\subsection{Organization}
The remainder of the paper is organized as follows.  First, we give some
definitions and useful facts in Section~\ref{sec:preliminaries}.  We give a
sublinear-time algorithm for approximating heat kernel pagerank in
Section~\ref{sec:hkprapprox}.  In Section~\ref{sec:goodcuts} we give the
analysis justifying our local clustering algorithm, which we present in
Section~\ref{sec:localpartition}.  Sections~\ref{sec:rankings}
and~\ref{sec:expresults} contain experimental results.  In both sections,
experiments are performed on real data and on synthetic graphs generated with
random graph generators.  In Section~\ref{sec:rankings} we demonstrate how the
rankings obtained using approximate heat kernel pagerank vectors are compared
with rankings obtained using exact heat kernel pagerank vectors.  In
Section~\ref{sec:expresults} we compute local clusters by implementing the
\partitionalg~algorithm.  We compare the volume and Cheeger ratio of these
clusters to those output by two existing sweep-based local clustering
algorithms.  The first is by a sweep of an exact heat kernel
pagerank~\cite{chung:partitionhkpr:im09} to compare the effects of heat kernel
pagerank computation, and the second by a PageRank
vector~\cite{acl:prgraphpartition:focs06}.  PageRank has a similar expression as
heat kernel pagerank except PageRank is a geometric sum whereas heat kernel
pagerank can be viewed as an exponential sum.  We expect better convergence
rates from heat kernel (see Section~\ref{sec:heatkernelandhkpr}).

\section{Preliminaries}
\label{sec:preliminaries}
Let $G = (V,E)$ be an undirected graph on $n$ vertices and $m$ edges.  We use $u
\sim v$ to denote $\{u,v\} \in E$.  The \emph{degree}, $d_v$, of a vertex $v$ is
the number of vertices $u$ such that $u\sim v$.  The \emph{volume} of a set of
vertices $S \subseteq V$ is the total degree of its vertices, $\vol(S) = \sum_{v
\in S} d_v$, and the \emph{edge boundary} of $S$ is the set of edges with one
vertex in $S$ and the other outside of $S$, $\partial(S) = \{ u\sim v ~:~ u \in
S, v \notin S \}$. When discussing the full vertex set, $V$, we write
$S\subseteq G$ and $\vol(G) = \vol(V)$.

Let $f\in \R^n$ be a row vector over the vertices of $G$.  Then the support of
$f$ is the set of vertices with nonzero values in $f$, $\supp(f) = \{u \in V~:~
f(u)\neq 0\}$.  For a subset of vertices $S$, we define $f(S) = \sum_{u\in S}
f(u)$.

\subsection{A local Cheeger inequality}
The quality of a cut can be measured by the ratio of the number of edges between
the two parts of the cut and the volume of the smaller side of the cut.  This is
called the \emph{Cheeger ratio} of a set, defined by
\begin{equation*}
\Phi(S) = \frac{ |\partial(S)|  }{ \min(\vol(S), \vol(V\setminus S)) }.
\end{equation*}
The \emph{Cheeger constant} of a graph is the minimal Cheeger ratio,
\begin{equation*}
\Phi(G) = \min_{S\subset G}\Phi(S).
\end{equation*}
Finally, for a given subset $S$ of a graph $G$, the \emph{local Cheeger ratio}
is defined
\begin{equation*}
\localcheeg = \min\limits_{T \subseteq S}\Phi(T).
\end{equation*}

Our local clustering algorithm is derived from a local version of the usual
Cheeger inequalities which relate the Cheeger constant of a graph to an
eigenvalue associated to the graph.  Namely, let the normalized Laplacian of a
graph be the matrix $\L = D^{-1/2}(D-A)D^{-1/2}$, where $D$ is the diagonal
matrix of vertex degrees and $A$ is the unweighted, symmetric adjacency matrix.
Also, let $\L_S$ be determined by a subset $S$ of size $|S|=s$ and define $\L_S
= D_S^{-1/2}(D_S-A_S)D_S^{-1/2}$ where $D_S$ and $A_S$ are the restricted
matrices of $D$ and $A$ with rows and columns indexed by vertices in $S$.  Then
the eigenvalues $\lambda_S := \lambda_{S,1} \leq \lambda_{S,2} \leq \cdots \leq
\lambda_{S,s}$ of $\L_S$ are also known as the \emph{Dirichlet eigenvalues of
$S$}, and are related to $\localcheeg$ by the following local Cheeger
inequality~\cite{chung:partitionhkpr:im09}:
\begin{equation}\label{eq:cheegerinequality}
\frac{1}{2}(\localcheeg)^2 \leq \lambda_S \leq \localcheeg.
\end{equation}

The inequality (\ref{eq:cheegerinequality}) will be used to derive a
relationship between a ranking according to heat kernel pagerank and sets with
good Cheeger ratios.  Details will be given in Section~\ref{sec:goodcuts}.

\subsection{Heat kernel and heat kernel pagerank}
\label{sec:heatkernelandhkpr}
The \emph{heat kernel pagerank} vector has entries indexed by the vertices of
the graph and involves two parameters; a non-negative real value $t$,
representing the temperature, and a preference row vector $f: V \rightarrow \R$,
by the following equation:
\begin{equation}\label{eq:hkpr}
\rho_{t,f} = e^{-t} \sum\limits_{k=0}^{\infty} \frac{t^k}{k!}fP^k
\end{equation}
where $P$ is the transition probability matrix
\begin{equation*}
(P)_{uv} =
\begin{cases}
1/d_u & \mbox{ if } u\sim v\\
0 & \mbox{ otherwise}.
\end{cases}
\end{equation*}

When $f$ is a probability distribution, the heat kernel pagerank can be regarded
as the expected distribution of a random walk according to the transition
probability matrix $P$.  A starting distribution we will be particularly
concerned with is that with all probability initially on a single vertex $u$,
i.e. $f = \chi_u$ where $\chi_u$ is the indicator vector for vertex $u$.  We
will denote the heat kernel pagerank vector over this distribution by $\rho_{t,u} :=
\rho_{t,\chi_u}$ and refer to $u$ as the \emph{seed} vertex.

The \emph{heat kernel} of a graph is defined $H_t = e^{-t\Delta}$ where $\Delta$
is the Laplace operator $\Delta = I - P$.  Then an alternative definition for
heat kernel pagerank is $\rho_{t,f} = fH_t$, and we have that heat kernel pagerank
satisfies the heat differential equation
\begin{equation}\label{eq:heateq}
\frac{\partial}{\partial t}\rho_{t,f} = -\rho_{t,f}(I-P).
\end{equation}

We can compare the heat kernel pagerank to the personalized PageRank vector,
given by
\begin{equation}\label{eq:pagerank}
\pr_{\alpha, f} = \alpha \sum\limits_{k=0}^{\infty}(1-\alpha)^kfP^k.
\end{equation}
In this definition, $\alpha$ is often called the \emph{jumping} or \emph{reset}
constant, meaning that at any step the random walk may jump to a vertex taken
from $f$ with probability $\alpha$.  When $f = \chi_u$ for some $u$, i.e.
preference is given to a single vertex, the random walk is ``reset" to the first
vertex of the walk, $u$, with probability $\alpha$.  We note that, compared to
the personalized PageRank vector, which can be viewed as a geometric sum, we can
expect better convergence rates from the heat kernel pagerank, defined as an
exponential sum.

\section{Heat Kernel Pagerank Approximation}
\label{sec:hkprapprox}
We begin our discussion of heat kernel pagerank approximation with an
observation.  Each term in the infinite series defining heat kernel pagerank in
$(\ref{eq:hkpr})$ is of the form $e^{-t}\frac{t^k}{k!}fP^k$ for
$k\in[0,\infty]$.  The vector $fP^k$ is the distribution after $k$ random walk
steps with starting distribution $f$.  Then, if we perform $k$ steps of a random
walk given by transition probability matrix $P$ from starting distribution $f$
with probability $p_k = e^{-t}\frac{t^k}{k!}$, the heat kernel pagerank vector
can be viewed as the expected distribution of this process.

This suggests a natural way to approximate the heat kernel pagerank.  That is,
we can obtain a close approximation to the expected distribution with
sufficiently many samples.  Our algorithm operates as follows.  We perform $r$
random walks to approximate the infinite sum, choosing $r$ large enough to bound
the error.  We also use the fact that very long walks are performed with small
probability.  As such, we limit the lengths of our random walks by a finite
number $K$.  Both $r,K$ depend on a predetermined error bound $\epsilon$.

In our analysis we will use the following definition of an
$\epsilon$-approximate vector.

\begin{definition}
\label{def:eps-approx}
Let $G$ be a graph on $n$ vertices, and let $f:V\rightarrow \R$ be a vector over
the vertices of $G$.  Let $\rho_{t,f}$ be the heat kernel pagerank vector according
to $f$ and $t$.  Then we say that $\nu \in \R^n$ is an
\emph{$\epsilon$-approximate vector} of $\rho_{t,f}$ if
\begin{enumerate}
\item for every vertex $v \in V$ in the support of $\nu$, 
$$(1-\epsilon)\rho_{t,f}(v) -\epsilon \leq \nu(v) \leq (1+\epsilon)\rho_{t,f}(v),$$
\item for every vertex with $\nu(v) = 0$, it must be that $\rho_{t,f}(v) \leq \epsilon$.
\end{enumerate}
\end{definition}

We note that this is a rather coarse requirement for an approximation, but
satisfies our needs for local clustering.  In the following algorithm, we
approximate $\rho_{t,u}$ by an $\epsilon$-approximate vector which we denote by
$\hat{\rho}_{t,u}$.  The running time of the algorithm is $\hkprcomplexity$.  The
method and complexity of the algorithm, $\hkprseedalg$, are similar to the
ApproxRow algorithm for personalized PageRank given
in~\cite{bbct:sublinearpr:waw12}.

\begin{algorithm}[H]
\floatname{algorithm}{}
\caption*{\hkprseedalgparams}
\label{alg:hkprseed}
\algblock[Name]{Start}{End}
input: a graph $G$, $t\in \R^{+}$, seed vertex $u\in V$, error parameter $0 < \epsilon < 1$.\\
output: $\rho$, an $\epsilon$-approximation of $\rho_{t,u}$.\\
\begin{algorithmic}
  \State initialize a $0$-vector $\rho$ of dimension $n$, where $n=|V|$
  \State $r \gets \rparam$
  \State $K \gets c\cdot\kparam$ for some choice of contant $c$\\
  \For {$r$ iterations}
    \Start
      \State simulate a $P$ random walk from vertex $u$ where $k$ steps are taken with probability $e^{-t}\frac{t^k}{k!}$ and $k \leq K$
      \State let $v$ be the last vertex visited in the walk
      \State $\rho[v] \gets \rho[v] + 1$
    \End
  \EndFor\\
  \State
  \Return $1/r \cdot \rho$  
\end{algorithmic} 
\end{algorithm}

\begin{theorem}\label{thm:hkpraccuracy}
Let $G$ be a graph and let $u$ be a vertex of $G$.  Then, the algorithm
\hkprseedalgparams outputs an $\epsilon$-approximate vector $\hat{\rho}_{t,u}$ of the
heat kernel pagerank $\rho_{t,u}$ for $0 < \epsilon < 1$ with probability at least
$1-\epsilon$.  The running time of \hkprseedalg~is $\hkprcomplexity$.
\end{theorem}

\subsection{Analysis of the heat kernel pagerank algorithm}
\label{sec:hkpranalysis}
Our analysis relies on the usual Chernoff bounds as stated below.

\begin{lemma}[\cite{bbct:sublinearpr:waw12}]\label{lem:chernoff}
Let $X_i$ be independent Bernoulli random variables with $X = \sum\limits_{i = 1}^r X_i$.  Then, 
\begin{enumerate}
\item for $0 < \epsilon < 1$, $\P(X < (1-\epsilon)r\E(X)) < \exp(-\frac{\epsilon^2}{2}r\E(X))$
\item for $0 < \epsilon < 1$, $\P(X > (1+\epsilon)r\E(X)) < \exp(-\frac{\epsilon^2}{4}r\E(X))$
\item for $c \geq 1$, $\P(X > (1+c)r\E(X)) < \exp(-\frac{c}{2}r\E(X))$.
\end{enumerate}
\end{lemma}

\begin{proof}[Theorem~\ref{thm:hkpraccuracy}]
Consider the random variable which takes on value $fP^k$ with probability $p_k =
e^{-t}\frac{t^k}{k!}$ for $k\in [0, \infty)$.  The expectation of this random
variable is exactly $\rho_{t,f}$.  Heat kernel pagerank can be understood as a series
of distributions of weighted random walks over the vertices, and the weights are
related to the number of steps taken in the walk.  The series can be computed by
simulating this process, i.e., draw $k$ according to $p_k$ and compute $fP^k$
with sufficiently many random walks of length $k$.

We approximate the infinite sum by limiting the walks to at most $K$ steps.  We
will take $K$ to be $K= \kparam$.  These interrupts risk the loss of
contribution to the expected value, but can be upper bounded by
$\frac{e^{-t}t^K}{K!} \leq \frac{\epsilon}{2}$ provided that $t > K/\log K$.
This is within the error bound for an approximate heat kernel pagerank.  If $t
\leq K/ \log K$, the expected length of the random walk is
\begin{align*}
\sum_{k=0}^{\infty} \frac{e^{-t}t^k}{k!} \cdot k = t < K/\log K.
\end{align*}
Thus we can ignore walks of length more than $K$ while maintaining $\rho_{t,u}(v) -
\epsilon \leq \hat{\rho}_{t,u}(v) \leq \rho_{t,u}(v)$ for every vertex $v$.

Next we show how many samples are necessary for our approximation vectors.  For
$k\leq K$, our algorithm simulates $k$ random walk steps with probability
$e^{-t}\frac{t^k}{k!}$.  To be specific, for a fixed $u$, let $X^v_k$ be the
indicator random variable defined by $X^v_k = 1$ if a random walk beginning from
vertex $u$ ends at vertex $v$ in $k$ steps.  Let $X^v$ be the random variable
that considers the random walk process ending at vertex $v$ in \emph{at most}
$k$ steps.  That is, $X^v$ assumes the vector $X^v_k$ with probability
$e^{-t}\frac{t^k}{k!}$.  Namely, we consider the combined random walk
\begin{equation*}
X^v=\sum_{k \leq K} e^{-t}\frac{t^k}{k!} X^v_k.
\end{equation*}

Now, let $\rho(k)_{t,u}$ be the contribution to the heat kernel pagerank vector
$\rho_{t,u}$ of walks of length at most $k$.  The expectation of each $X^v$ is
$\rho(k)_{t,u}(v)$.  Then, by Lemma~\ref{lem:chernoff},
\begin{align*}
\P(X^v < (1-\epsilon) \rho(k)_{t,u}(v)\cdot r) &< \exp(-\rho(k)_{t,u}(v)r\epsilon^2/2)\\
&= \exp(-(8/\epsilon)\rho(k)_{t,u}(v)\log n)\\
&< n^{-4}
\end{align*}
for every component with $\rho_{t,u}(v) > \epsilon$, since then $\rho(k)_{t,u}(v) >
\epsilon/2$.  Similarly,
\begin{align*}
\P(X^v > (1+\epsilon) \rho(k)_{t,u}(v)\cdot r) &< \exp(-\rho(k)_{t,u}(v)r\epsilon^2/4)\\
&= \exp(-(4/\epsilon)\rho(k)_{t,u}(v)\log n)\\
&< n^{-2}.
\end{align*}
We conclude the analysis for the support of $\rho_{t,u}$ by noting that $\hat{\rho}_{t,u}
= \frac{1}{r} X^v$, and we achieve an $\epsilon$-multiplicative error bound for
every vertex $v$ with $\rho_{t,u}(v) > \epsilon$ with probability at least
$1-O(n^{-2})$.

On the other hand, if $\rho_{t,u}(v)\leq\epsilon$, by the third part of
Lemma~\ref{lem:chernoff}, $\P(\hat{\rho}_{t,u}(v) > 2\epsilon) \leq
n^{-8/\epsilon^2}$.  We may conclude that, with high probability, $\hat{\rho}_{t,u}(v)
\leq 2\epsilon$.

For the running time, we use the assumptions that performing a random walk step
and drawing from a distribution with bounded support require constant time.
These are incorporated in the random walk simulation, which dominates the
computation.  Therefore, for each of the $r$ rounds, at most $K$ steps of the
random walk are simulated, giving a total of $rK =  O\Big(\rparam \cdot
\kparam\Big)= \tilde {O}(1)$ queries.
\qed\end{proof}

\begin{remark}
This bound on $K$ is not tight.  However, it is enough to use $cK$ for some
small constant $c$ to cluster vertices with $\epsilon$-approximate heat kernel
pagerank vectors computed with bounded random walks.  Regardless, this value
$O(K)$ is independent of the size of the graph and never affects the running
time.  See Section~\ref{sec:rankings} for a futher discussion.
\end{remark}

\begin{remark}
We note that the algorithm works for any $t$, but a good choice of $t$ will be
related to the size of the local cluster $S$ and a desirable convergence rate.
In particular, the constraints put on $t$ are necessary for our local clustering
results, presented in Section~\ref{sec:goodcuts}.
\end{remark}

The algorithm for efficient heat kernel pagerank computation has promise for a
variety of applications.  It has been shown in~\cite{cs:hklinear:13} how to
apply heat kernel pagerank in solving symmetric diagonally dominant linear
systems with a boundary condition, for example.

\section{Finding Good Local Cuts}
\label{sec:goodcuts}
The premise of the local clustering algorithm is to find a good cut near a
specified vertex by performing a \emph{sweep} over a vector associated to that
vertex, which we will specify.  Let $p:V\rightarrow \R$ be a probability
distribution vector over the vertices of the graph of support size $N_p =
\supp(p)$.  Then, consider a \emph{probability-per-degree} ordering of the
vertices where $p(v_1)/d_{v_1} \geq p(v_2)/d_{v_2} \geq \cdots \geq
p(v_{N_p})/d_{v_{N_p}}$.  Let $S_i$ be the set of the first $i$ vertices per the
ordering.  We call each $S_i$ a \emph{segment}.  Then the process of
investigating the cuts induced by the segments to find an optimal cut is called
performing a sweep over $p$.  

In this section we will show how a sweep over a single heat kernel pagerank
vector finds local cuts.  Specifically, we show that for a subset $S$ with
$\vol(S) \leq \vol(G)/4$ and $\Phi(S) \leq \phi$, and for a large number of
vertices $u\in S$, performing a sweep over the vector $\hat{\rho}_{t,u}$, where
$\hat{\rho}_{t,u}$ is an $\epsilon$-approximation of $\rho_{t,u}$, will find a set with
Cheeger ratio at most $O(\sqrt{\phi})$.

\begin{remark}
Though all the vertices in the support of the vector are sorted to build
segments, in practice the sweep will be aborted after the volume of the current
segment is larger than the target size.  This is the \emph{locality} of the
algorithm, and ensures that the amount of work performed is proportional to the
volume of the output set.
\end{remark}

The $\varsigma$-local Cheeger ratio of a sweep over a vector $\nu$ is the
minimum Cheeger ratio over segments $S_i$ with volume $0 \leq \vol(S_i)\leq
2\varsigma$.  Let $\Phi_{\varsigma}(\nu)$ the $\varsigma$-local Cheeger ratio of
cuts over a sweep of $\nu$ that separates sets of volume between $0$ and
$2\varsigma$.

We will use the following bounds for heat kernel pagerank in terms of local
Cheeger ratios and sweep cuts to reason that many vertices $u$ can serve as good
seeds for performing a sweep.

\begin{lemma}
\label{thm:volumegoodset}
Let $G$ be a graph and $S$ a subset of vertices of volume $\varsigma \leq
\rm{vol}(G)/4$.  Then the set of $u\in S$ satisfying
\begin{equation*}
\frac{1}{2}e^{-t\localcheeg} 
\leq \rho_{t,u}(S)
\leq \sqrt{\varsigma}e^{-t\Phi_{\varsigma}(\rho_{t,f_S})^2/4}
\end{equation*}
has volume at least $\varsigma/2$.
\end{lemma}

To proof Lemma~\ref{thm:volumegoodset}, we begin with some bounds for heat
kernel pagerank in terms of local Cheeger ratios and sweep cuts.  For a subset
$S$, define $f_S$ to be the following distribution over the vertices,
\begin{equation*}
f_S(u) =
\begin{cases}
d_u/\vol(S) & \mbox{ if }u\in S\\
0 & \mbox{ otherwise.}
\end{cases}
\end{equation*}
Then the expected value of $\rho_{t,u}(S)$ over $u$ in $S$ is given by:
\begin{align}
\E(\rho_{t,u}(S)) &= \sum\limits_{u\in S}\frac{d_u}{\vol(S)}\rho_{t,u}(S)\nonumber\\
&= \sum\limits_{u\in S}\frac{d_u}{\vol(S)}(\chi_u H_t)(S)\nonumber\\
&= f_S H_t(S)\nonumber\\
&= \rho_{t,f_S}(S).\label{eq:ehkpru}
\end{align}

We will make use of the following result, given here without proof
(see~\cite{chung:partitionhkpr:im09}), which bounds the expected value of
$\rho_{t,u}(S)$ given by (\ref{eq:ehkpru}) in terms of local Cheeger ratios.

\begin{lemma}[\cite{chung:partitionhkpr:im09}]
\label{lem:lowerbound}
In a graph $G$, and for a subset $S$, the following holds:
\[
\frac{1}{2}e^{-t\localcheeg} \leq \frac{1}{2}e^{-t\lambda_S} \leq \rho_{t,f_S}(S).
\]
\end{lemma}

Next, we use an upper bound on the amount of probability remaining in $S$ after
sufficient mixing.  This is an extension of a theorem given
in~\cite{chung:partitionhkpr:im09}.

\begin{theorem}
\label{thm:eupperbound}
Let $G$ be a graph and $S$ a subset of vertices with volume $\varsigma \leq
\rm{vol}(G)/4$.  Then,
\[ \rho_{t,f_S}(S) \leq \sqrt{\varsigma}e^{-t\Phi_{\varsigma}(\rho_{t,f_S})^2/4}.  \]
\end{theorem}

To prove Theorem~\ref{thm:eupperbound}, we define the following for an arbitrary
function $f:V\rightarrow \R$ and any integer $x$ with $0\leq x \leq \vol(G)/2$,
\begin{equation*}
f(x) = \max\limits_{T\subseteq V\times V, |T|=x} \sum\limits_{(u,v)\in T} f(u,v), ~~~
f(u,v)=
\begin{cases}
f(u)/d_u, \mbox{ if $u \sim v$,}\\
0, \mbox{ otherwise}.
\end{cases}
\end{equation*}
The above definition can be extended to all real values of $x$,
\begin{equation*}
f(x) = \max\limits_{T\subseteq V\times V, |T|=x} \sum\limits_{(u,v)\in T} \alpha_{uv}f(u,v), ~~~
\alpha_{uv} \leq 1 \mbox{ if }u\sim v, ~~\sum_{u\sim v}\alpha_{uv} = x.
\end{equation*}

\begin{claim}\label{claim:f}
Let $S_i$ be a segment according to a vector $f:V\rightarrow\R$ such that
$x=\rm{vol}(S_i)$ and $f(v) > 0$ for every $v\in S_i$.  Then $$f(x) = \sum_{u\in
S_i}f(u) = f(S_i).$$
\end{claim}

\begin{proof}
We are considering the maximum over a subset of vertex pairs $T$ of size
$\vol(S_i)$.  Since we are only adding values over vertex pairs which are edges
in $G$, this maximum is achieved when
\begin{align*}
f(x) &= \sum\limits_{u\in S_i}\sum\limits_{v\sim u}f(u)/d_u\\
&=\sum\limits_{u\in S_i} f(u) \sum\limits_{v\sim u}1/d_u\\
&= f(S_i).
\end{align*}
\qed\end{proof}

\begin{proof}[Theorem~\ref{thm:eupperbound}]
Let $Z$ be the lazy random walk $Z = 1/2(I + P)$.  Then,
\begin{align*}
fZ(S) &= 1/2\Big( f(S) + \sum\limits_{u\sim v \in S} f(u,v) \Big)\\
&= 1/2\Big( \sum\limits_{u \vee v \in S} f(u,v) + \sum\limits_{u \wedge v \in S} f(u,v) \Big)\\
&\leq 1/2(f(\vol(S) + |\partial(S)|) + f(\vol(S) - |\partial(S)|))\\
&= 1/2\Big(f(\vol(S)(1 + \Phi(S))) + f(\vol(S)(1 - \Phi(S)))\Big).
\end{align*}

Let $f_t = \rho_{t,f_S}$, and let $x$ satisfy $0 \leq x \leq 2\varsigma \leq
\vol(G)/2$ and represent a volume of some set $S_i$.  Then taking cue from the
above inequality, we can associate $S$ to $S_i$, $\vol(S)$ to $\vol(S_i) = x$
and $\Phi(S)$ to the minimum Cheeger ratio of a set $S_i$ satisfying $\vol(S_i)
= x \leq 2\varsigma$, or $\Phi_{\varsigma}(\rho_{t,f_S})$.  Then using Claim~\ref{claim:f},
\[ f_tZ(x) \leq 1/2(f_t(x(1+\Phi_{\varsigma}(\rho_{t,f_S}))) + f_t(x(1-\Phi_{\varsigma}(\rho_{t,f_S})))).  \]
Now consider the following differential inequality,
\begin{align}
\frac{\partial}{\partial t}f_t(x) &= -\rho_{t,f_S}(I-W)(x)\label{heatdiff}\\
&= -2\rho_{t,f_S}(I-Z)(x)\nonumber\\
&= -2f_t(x) + 2f_tZ(x)\nonumber\\
&\leq -2f_t(x) + f_t(x(1+\Phi_{\varsigma}(\rho_{t,f_S})))\nonumber\\
&~~~~+ f_t(x(1-\Phi_{\varsigma}(\rho_{t,f_S})))\label{diffin}\\
&\leq 0.\label{concave}
\end{align}
Line (\ref{heatdiff}) follows from (\ref{eq:heateq}), and line (\ref{concave})
follows from the concavity of $f$.

Consider $g_t(x)$ to be $g_t(x) = \sqrt{x}e^{-t\Phi_{\varsigma}(\rho_{t,f_S})^2/4}$.  Then,
\begin{align}
&-2g_t(x) + g_t(x(1+\Phi_{\varsigma}(\rho_{t,f_S}))) + g_t(x(1-\Phi_{\varsigma}(\rho_{t,f_S})))\nonumber\\
&= -2g_t(x) + \sqrt{1+\Phi_{\varsigma}(\rho_{t,f_S})}g_t(x) + \sqrt{1-\Phi_{\varsigma}(\rho_{t,f_S})}g_t(x)\nonumber\\
&= (-2 + \sqrt{1+\Phi_{\varsigma}(\rho_{t,f_S})} + \sqrt{1-\Phi_{\varsigma}(\rho_{t,f_S})})g_t(x)\nonumber\\
&\leq \frac{-\Phi_{\varsigma}(\rho_{t,f_S})^2}{4} g_t(x)\label{yineq}\\
&= \frac{\partial}{\partial t}g_t(x),\nonumber
\end{align}
where we use the fact that $-2 + \sqrt{1+y} + \sqrt{1-y} \leq -y^2/4$ for
$y\in(0,1]$ in line (\ref{yineq}).  Now, since $f_t(0) = g_t(0)$ and
$\frac{\partial}{\partial t}f_t(x)|_{t=0} \leq \frac{\partial}{\partial
t}g_t(x)|_{t=0}$, 
\begin{align*}
&-2f_t(x) + f_t(x(1+\Phi_{\varsigma}(\rho_{t,f_S}))) + f_t(x(1-\Phi_{\varsigma}(\rho_{t,f_S})))\\
&< -2g_t(x) + g_t(x(1+\Phi_{\varsigma}(\rho_{t,f_S}))) + g_t(x(1-\Phi_{\varsigma}(\rho_{t,f_S}))),
\end{align*}
and in particular, $\frac{\partial}{\partial t}f_t(x) \leq
\frac{\partial}{\partial t}g_t(x)$.  Since $f_0(x) \leq g_0(x)$, we may conclude
that
\begin{equation*}
f_t(x) \leq g_t(x) = \sqrt{x}e^{-t\Phi_{\varsigma}(\rho_{t,f_S})^2/4}.
\end{equation*}
\qed\end{proof}

Using Lemma~\ref{lem:lowerbound} and Theorem~\ref{thm:eupperbound}, we arrive at
the following useful inequalities.

\begin{corollary}
\label{cor:bounds}
Let $G$ be a graph and $S$ a subset with volume $\varsigma \leq \rm{vol}(G)/4$.
Then,
\begin{equation*}
\frac{1}{2}e^{-t\localcheeg} 
\leq \rho_{t,f_S}(S)
\leq \sqrt{\varsigma}e^{-t\Phi_{\varsigma}(\rho_{t,f_S})^2/4}.
\end{equation*}
\end{corollary}

We are now prepared to prove Lemma~\ref{thm:volumegoodset}.

\begin{proof}[Lemma~\ref{thm:volumegoodset}]
Let $F$ be the set of seeds $F = \{u\in S ~:~ \rho_{t,u}(S) \leq 2\rho_{t,f_S}(S)\}$.
Then, by (\ref{eq:ehkpru}),
\begin{equation*}
F = \{u\in S ~:~ \rho_{t,u}(S) \leq 2\E(\rho_{t,u}(S))\}.
\end{equation*}

Now we consider the set of vertices not included in $F$,
\begin{align*}
\E(\rho_{t,u}(S) ~|~ u\notin F) &\geq \sum\limits_{u\notin F} \frac{d_u}{\vol(S)} 2\E(\rho_{t,u}(S))\\
&\geq \frac{\vol(S\setminus F)}{\vol(S)} 2\sum\limits_{u\notin F} \E(\rho_{t,u}(S)).
\end{align*}
Which implies
\begin{equation*}
\frac{\vol(S)}{2} > \vol(S\setminus F) ~~~\mbox{ or, }~~~ \vol(F) > \varsigma/2.
\end{equation*}
\qed\end{proof}

We can use Lemma~\ref{thm:volumegoodset} to reason that many vertices $u$
satisfy the above inequalities, and thus can serve as good seeds for performing
a sweep.

\subsection{A local graph clustering algorithm}
\label{sec:localpartition}
It follows from Lemma~\ref{thm:volumegoodset} that the ranking induced by a
heat kernel pagerank vector with appropriate seed vertex can be used to find a
cut with approximation guarantee $O(\sqrt{\phi})$ by choosing the appropriate
$t$.  To obtain a sublinear time local clustering algorithm for massive graphs,
we use \hkprseedalg~to efficiently compute an $\epsilon$-approximate heat kernel
pagerank vector, $\hat{\rho}_{t,u}$, to rank vertices.

The ranking induced by $\hat{\rho}_{t,u}$ is not very different from that of a true
vector $\rho_{t,u}$ in the support of $\hat{\rho}_{t,u}$ (for an experimental analysis,
see Section~\ref{sec:rankings}).  Namely, using the bounds of
Lemma~\ref{lem:lowerbound}, we have
\[ \hat{\rho}_{t,u}(S) \geq (1-\epsilon)\rho_{t,u}(S)
-\epsilon s,  \] 
where $s = |S|$.  In particular,
\begin{equation}\label{eq:approxineq}
\frac{1}{2}(1-\epsilon)e^{-t\localcheeg} -\epsilon s 
\leq \hat{\rho}_{t,u}(S) 
\leq \sqrt{\varsigma}e^{-t\Phi_{\varsigma}(\hat{\rho}_{t,u})^2/4}
\end{equation}
for a set of vertices $u$ of volume at least $\varsigma/2$.

\begin{theorem}\label{thm:main}
Let $G$ be a graph and $S\subset G$ a subset with $\rm{vol}(S)$$= \varsigma
\leq$$\rm{vol}(G)/4$, $|S|=s$, and Cheeger ratio $\Phi(S) \leq \phi$.  Let
$\hat{\rho}_{t,u}$ be an $\epsilon$-approximate of $\rho_{t,u}$ for some vertex $u\in S$.
Then there is a subset $S_t\subset S$ with $\rm{vol}(S_t)$$\geq\varsigma/2$ for
which a sweep over $\hat{\rho}_{t,u}$ for any vertex $u\in S_t$ with
\begin{enumerate}
\item$t=\tparamcluster$ and
\item$\Phi_{\varsigma}(\hat{\rho}_{t,u})^2 \leq 4/t\log(2)$\label{assumption:sweepcheeg}
\end{enumerate}
finds a set with $\varsigma$-local Cheeger ratio at most $\sqrt{8\phi}$.
\end{theorem}

\begin{proof}
Let $u$ be a vertex in $S_t$ as described in the theorem statement.  Using the
inequalities (\ref{eq:approxineq}), we can bound the local Cheeger ratio by a
sweep over $\hat{\rho}_{t,u}$:
\begin{equation*}
e^{-t\localcheeg} \leq \frac{2}{1-\epsilon}(\sqrt{\varsigma}e^{-t\Phi_{\varsigma}(\hat{\rho}_{t,u})^2/4}+\epsilon s)
\end{equation*}
which implies
\begin{equation*}
e^{-t\localcheeg} \leq e^{-t\Phi_{\varsigma}(\hat{\rho}_{t,u})^2/4} \Big(\frac{2\sqrt{\varsigma}}{1-\epsilon} + \epsilon s e^{t\Phi_{\varsigma}(\hat{\rho}_{t,u})^2/4}\Big),
\end{equation*}
and by the assumption \ref{assumption:sweepcheeg}, we have
\begin{align*}
e^{-t\localcheeg} &\leq e^{-t\Phi_{\varsigma}(\hat{\rho}_{t,u})^2/4} \Big(\frac{2\sqrt{\varsigma}}{1-\epsilon} + 2\epsilon s\Big)\\
\localcheeg &\geq \frac{\Phi_{\varsigma}(\hat{\rho}_{t,u})^2}{4} - \frac{\log(\frac{2\sqrt{\varsigma}}{1-\epsilon} + 2\epsilon s)}{t}.
\end{align*}

Let $x=\log(\frac{2\sqrt{\varsigma}}{1-\epsilon} + 2\epsilon s)$.
Then, 
\begin{equation*}
\Phi_{\varsigma}(\rho_{t,f_S})^2 \leq 4\localcheeg + 4x/t.
\end{equation*}
Since $\localcheeg \leq \Phi(S) \leq \phi$ and $t = \phi^{-1}x$, it follows that
$\Phi_{\varsigma}(\hat{\rho}_{t,u}) \leq \sqrt{8\phi}$.  In particular, a sweep over
$\hat{\rho}_{t,u}$ finds a cut with Cheeger ratio $O(\sqrt{\phi})$ as long as $u$ is
contained in $S_t$.
\qed\end{proof}

We are now prepared to give our algorithm for finding cuts locally with heat
kernel pagerank.  The algorithm takes as input a starting vertex $u$, a desired
volume $\varsigma$ for the cut set, and a target Cheeger ratio $\phi$ for the
cut set.  Then, to find a set achieving a minimum $\varsigma$-local Cheeger
ratio, we perform a sweep over an approximate heat kernel pagerank vector with
the starting vertex as a seed.

\begin{algorithm}[H]
\floatname{algorithm}{}
\caption*{\partitionalgparams}
\label{alg:localpart}
input: a graph $G$, a vertex $u$, target cluster size $s$, target cluster volume
$\varsigma \leq \vol(G)/4$, target Cheeger ratio $\phi$, error parameter
$\epsilon$.\\
output: a set $T$ with $\varsigma/2 \leq \vol(T) \leq 2\varsigma$,
$\Phi(T) \leq \sqrt{8\phi}$.\\
\begin{algorithmic}[1]
  \State $t \gets \tparamcluster$
  \State $\hat{\rho} \gets \hkprseedalgparams$\label{line:hkpr}
  \State sort the vertices of $G$ in the support of $\hat{\rho}$ according to the ranking
$\hat{\rho}[v]/d_{v}$\label{line:sort}
  \For{$j\in[1,n]$}\label{line:sweep}
    \State $S_j = \bigcup_{i\leq j}v_i$
    \If{$\vol(S_j) > 2\varsigma$}\label{line:volume}
      \State output \texttt{NO CUT FOUND}, break
    \ElsIf{$\varsigma/2 \leq \vol(S_j) \leq 2\varsigma$ and $\Phi(S_j) \leq \sqrt{8\phi}$}\label{line:checks}
      \State output $S_j$
    \Else
      \State output \texttt{NO CUT FOUND}
    \EndIf
  \EndFor
\end{algorithmic}
\end{algorithm}

\begin{theorem}{\label{thm:localpart}}
Let $G$ be a graph which contains a subset $S$ of volume at most $\rm{vol}(G)/4$
and Cheeger ratio bounded by $\phi$.  Further, assume that $u$ is contained in
the set $S_t\subseteq S$ as defined in Theorem~\ref{thm:main}.  Then
\partitionalgparams outputs a cutset $T$ with $\varsigma$-local Cheeger ratio at
most $\sqrt{8\phi}$.  The running time is the same as that of~\hkprseedalg.
\end{theorem}

\begin{proof}
Since it is given that $u\in S_t$ for $t=\tparamcluster$, and by the assumptions
on $G$ and $S$, Theorem~\ref{thm:main} states that a sweep over the approximate
heat kernel pagerank vector $\hat{\rho}$ will find a set with $\varsigma$-local
Cheeger ratio at most $\sqrt{8\phi}$.  The checks performed in
line~\ref{line:checks} of the algorithm discover such a set.

The computational work reduces to the main procedures of computing the heat
kernel pagerank vector in line~\ref{line:hkpr} and performing a sweep over the
vector in line~\ref{line:sweep}.  Performing a sweep involves sorting the
support of the vector (line~\ref{line:sort}) and calculating the conductance of
segments.  From the guarantees of an $\epsilon$-approximate heat kernel pagerank
vector, any vertex with average probability less than $\epsilon$ will be
excluded from the support.  Then the volume of a vector $\hat{\rho}$ output in
line~\ref{line:hkpr} is $O(\epsilon^{-1})$, and performing a sweep over
$\hat{\rho}$ can be done in $O(\epsilon^{-1}\log(\epsilon^{-1}))$ time.  The
algorithm is therefore dominated by the time to compute a heat kernel pagerank
vector, and the total running time is $\hkprcomplexity$.
\qed\end{proof}

\section{Ranking Vertices with Approximate Heat Kernel Pagerank}
\label{sec:rankings}
The backbone procedure of the local clustering algorithm is the sweep: ranking
the vertices of the graph according to their approximate heat kernel pagerank
values, and then testing the quality of the cluster obtained by adding vertices
one at a time in the order of this ranking.  To this end, in this section we
compare the rankings of vertices obtained using exact heat kernel pagerank
vectors with approximate heat kernel pagerank vectors.  Specifically, we
consider how accuracy changes as the upper bound of random walk lengths, $K$,
vary.

In the following experiments, we approximate heat kernel pagerank vectors by
sampling random walks of length $\min\{k, K\}$, where $k$ is chosen with
probability $p_k = e^{-t}\frac{t^k}{k!}$.  We test the values computed with
different values of $K$.  Since the expected value of a random walk length $k$
chosen with probability $p_k = e^{-t}\frac{t^k}{k!}$ is $t$, we set $K$ to range
from $1$ to approximately $t$ for a specified value of $t$.

In each trial, for a given graph, seed vertex, and value of $t$, we compute an
exact heat kernel pagerank vector $\rho_{t,u}$ and an approximate heat kernel
pagerank vector $\hat{\rho}_{t,u}$ using \hkprseedalg~but limiting the length of
random walks to $K$ for various $K$ as described above.  We then measure how
similar the vectors are in two ways.  First, we compare the vector values
computed.  Second, we compare the rankings obtained with each vector.  The
following are the measures used:

\begin{enumerate}
\item \emph{Comparing vector values.}  We measure the error of the approximate
vector $\hat{\rho}_{t,u}$ by examining the values computed for each vertex and
comparing to an exact vector $\rho_{t,u}$.  We use the following measures:
  \begin{itemize}
  \item \textbf{Average $L_1$ error}: The average absolute error over all vertices of the graph,
  \begin{equation}\label{eq:measure-l1err}
  \mbox{average $L_1$ error } := \frac{1}{n}\sum_{i=1}^n
|\rho_{t,u}(v_i)-\hat{\rho}_{t,u}(v_i)|.
  \end{equation}
  \item \textbf{$\epsilon$-error}: The accumulated error in excess of an
$\epsilon$-approximation (see Definition~\ref{def:eps-approx}),
  \begin{align}
  \mbox{$\epsilon$-error } := &\sum_{v\in V, \hat{\rho}_{t,u}(v) > 0} \max\{|\rho_{t,u}(v) - \hat{\rho}_{t,u}(v)|
- \epsilon\rho_{t,u}(v), 0\}\nonumber\\
  &+ \sum_{v\in V, \hat{\rho}_{t,u}(v) = 0} \max\{\rho_{t,u}(v) - \epsilon,
0\}.\label{eq:measure-epserr} \end{align}
  \end{itemize}
\item \emph{Comparing vector rankings.}  To measure the similarity of vertex
rankings we use the intersection difference (see
\cite{benzi2013total,fagin2003comparing} among others).  For a ranked list of
vertices $A$, let $A_i$ be the set of items with the top $i$ rankings.  Then we
use the following measures:
  \begin{itemize}
  \item \textbf{Intersection difference}: Given two ranked lists of vertices,
$A$ and $B$, each of length $n$, the intersection difference is
  \begin{equation}\label{eq:measure-isim}
  \mbox{intersection difference} := dist(A,B) = \frac{1}{n} \sum\limits_{i=1}^n
\frac{|A_i \oplus B_i|}{2i},
  \end{equation}
  where $A_i \oplus B_i$ denotes the symmetric difference $(A_i \setminus B_i) \cup (B_i
\setminus A_i)$.
  \item \textbf{Top-$k$ intersection difference}: The intersection difference 
among the top $k$ elements in each ranking,
  \begin{equation}\label{eq:measure-isimk}
  \mbox{top-$k$ intersection difference} := dist_k(A,B) = \frac{1}{k}
\sum\limits_{i=1}^k \frac{|A_i \oplus B_i|}{2i}.
  \end{equation}
  \end{itemize}
Intersection difference values lie in the range $[0,1]$, where a difference of
$0$ is achieved for identical rankings, and $1$ for totally disjoint lists.  In
these experiments, $A$ is the list of vertices ranked according to an exact heat
kernel pagerank vector $\rho_{t,u}$, and $B$ is the list of vertices ranked
according to an $\epsilon$-approximate heat kernel pagerank vector $\hat{\rho}_{t,u}$.
\end{enumerate}

In every trial we choose $t=\tparamcluster$ as specified in the local clustering
algorithm stated in Section~\ref{sec:localpartition}.  This value depends on
several parameters, including desired Cheeger ratio, cluster size, and cluster
volume.  Specifics are provided for each set of trials.

\subsection{Synthetic graphs}
\label{sec:synthranking}

\subsubsection{Random graph models}
In this series of trials we use three different models of random graph
generation provided in the NetworkX~\cite{networkx} Python package, which we
describe presently.

The first is the Watts-Strogatz small world model~\cite{wattsstrogatz},
generated with the command \texttt{connected_watts_strogatz} in NetworkX.  In
this model, a ring of $n$ vertices is created and then each vertex is connected
to its $d$ nearest neightbors.  Then, with probability $p$, each edge $(u,v)$ in
the original construction is replaced by the edge $(u,w)$ for a random existing
vertex $w$.  The model takes parameters $n, d, p$ as input.

The second is the preferential attachment (Barab\'{a}si-Albert)
model~\cite{barabasialbert}.  Graphs in this model are created by adding $n$
vertices one at a time, where each new vertex is adjoined to $d$ edges where
each edge is chosen with probability proportional to the degree of the
neighboring vertex.  This is generated with the \texttt{barabasi_albert_graph}
generator in NetworkX.  The model takes parameters $n, d$ as input.

The third NetworkX generator is \texttt{powerlaw_cluster_graph}, which uses the
Holme and Kim algorithm for generating graphs with powerlaw degree distribution
and approximate average clustering~\cite{holmekim}.  It is essentially the
Barab\'{a}si-Albert model, but each random edge forms a triangle with another
neighbor with probability $p$.  The model takes parameters $n, d, p$ as input.

Table~\ref{table:randomgraphs} lists the random graph models used and their
parameters.

\begin{table}
\begin{tabular}{|l|l|l|}
\hline
\multicolumn{1}{|c|}{Model} & \multicolumn{1}{|c|}{Source} & \multicolumn{1}{|c|}{Parameters}\\
\hline\hline
small world & Watts-Stragatz\cite{wattsstrogatz} & $n$, the size of the vertex set,\\
            &                & $d$, the number of neighbors each vertex is
assigned,\\
            &                & $p$, the probability of switching an edge.\\
\hline
preferential& Barab\'{a}si-Albert\cite{barabasialbert} & $n$, the size of the
vertex set,\\
attachment  &                                          & $d$, the number of
neighbors each vertex is assigned\\
\hline
powerlaw    & Holme and Kim~\cite{holmekim} & $n$, the size of the vertex set,\\
cluster     &                               & $d$, the number of neighbors each
vertex is assigned,\\
            &                               & $p$, the probability of forming a
triangle\\
\hline
\end{tabular}
\caption{Random graph models used.}
\label{table:randomgraphs}
\end{table}

\subsubsection{Procedure} For every value of $K$ that we test, we generate ten
random graphs using each of the three random graph models.  For each graph we
choose a random seed vertex $u$ with probability proportional to degree, and we
choose $t$ as $t=\tparamcluster$ according to the values in
Table~\ref{table:synthrankingparams}.  Then for each graph we compare an exact
heat kernel pagerank vector $\rho_{t,u}$ and the average of two
$\epsilon$-approximate heat kernel pagerank vectors $\hat{\rho}_{t,u}$.  The results
we present are the average over all trials for each $K$ and each type of graph.
We use $d=5$ and $p=0.1$ in every trial, and $n=100$ for the first set of trials
(Figure~\ref{fig:random_rank_100}) and $n=500$ for the second
(Figure~\ref{fig:random_rank_500}).  These parameters are outlined in
Table~\ref{table:synthrankingparams}.

\begin{table}
\centering
\begin{tabular}{|p{2cm}|c|c|c|c|c|c|c|c|}
\hline
\multicolumn{1}{|c|}{Model} & $|V|$ & ~~$d$~~ & ~~$p$~~ & ~~$\epsilon$~~ 
Target & Target & Target & $t$\\
 & & & & & Cheeger ratio & cluster size & cluster volume &\\
\hline\hline
small world  & $100$ & $5$ & $0.1$ & $0.1$ & $\phi = 0.05$ & $s = 100$ & $\varsigma = 500$ & $84.9$\\
             & $500$ & $5$ & $0.1$ & $0.1$ & $\phi = 0.05$ & $s = 100$ &
$\varsigma = 500$ & $84.9$\\\hline
preferential & $100$ & $5$ & - & $0.1$ & $\phi = 0.05$ & $s = 100$ & $\varsigma = 500$ & $84.9$\\
attachment   & $500$ & $5$ & - & $0.1$ & $\phi = 0.05$ & $s = 100$ &
$\varsigma = 500$ & $84.9$\\\hline
powerlaw     & $100$ & $5$ & $0.1$ & $0.1$ & $\phi = 0.05$ & $s = 100$ & $\varsigma = 500$ & $84.9$\\
cluster      & $500$ & $5$ & $0.1$ & $0.1$ & $\phi = 0.05$ & $s = 100$ &
$\varsigma = 500$ & $84.9$\\
\hline
\end{tabular}
\caption{Parameters used for random graph generation and to compute $t$ for vector computations.}
\label{table:synthrankingparams} 
\end{table}

\subsubsection{Discussion} For each graph and value of $K$, we measure the
$\epsilon$-error, the average $L_1$ error, the intersection difference and the
top-$10$ intersection difference of an approximate heat kernel pagerank vector
as compared to an exact heat kernel pagerank vector.
Figure~\ref{fig:random_rank_100} plots the above measures for graphs over
$n=100$ vertices, while Figure~\ref{fig:random_rank_500} plots these measures
for graphs over $n=500$ vertices.  In both Figures~\ref{fig:random_rank_100}
and~\ref{fig:random_rank_500}, each subplot charts a different notion of error
(from top left, clockwise: $\epsilon$-error, average $L_1$ error, intersection
difference and top-$10$ intersection difference) on the y-axis against $K$ on
the x-axis.  

\begin{figure}
\centering
\textbf{Trials on $100$-vertex random graphs.}
\includegraphics[width=\textwidth]{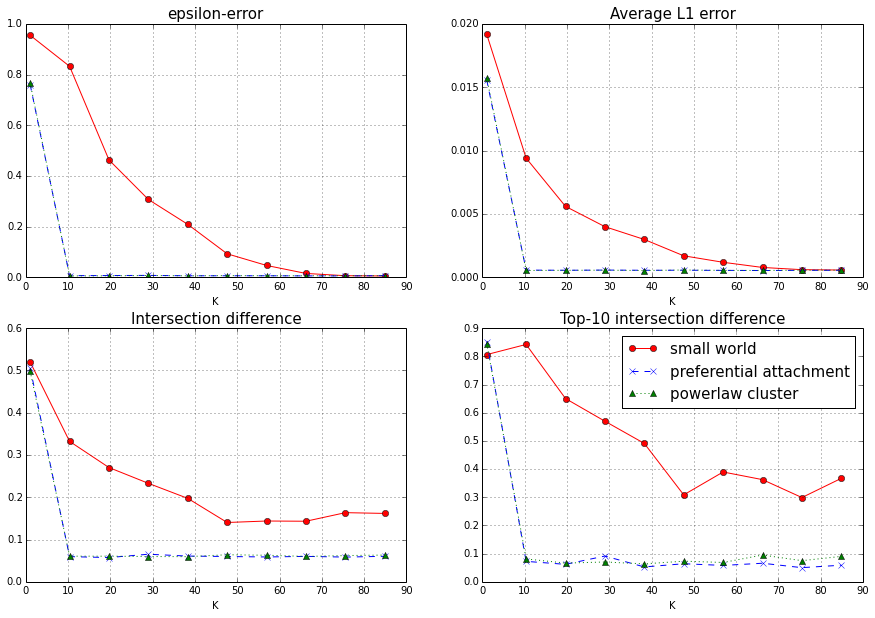}
\caption{Different measures of error for random graphs on $100$ vertices when approximating heat
kernel pagerank with varying random walk lengths.}
\label{fig:random_rank_100}
\end{figure}

\begin{figure}
\centering
\textbf{Trials on $500$-vertex random graphs.}
\includegraphics[width=\textwidth]{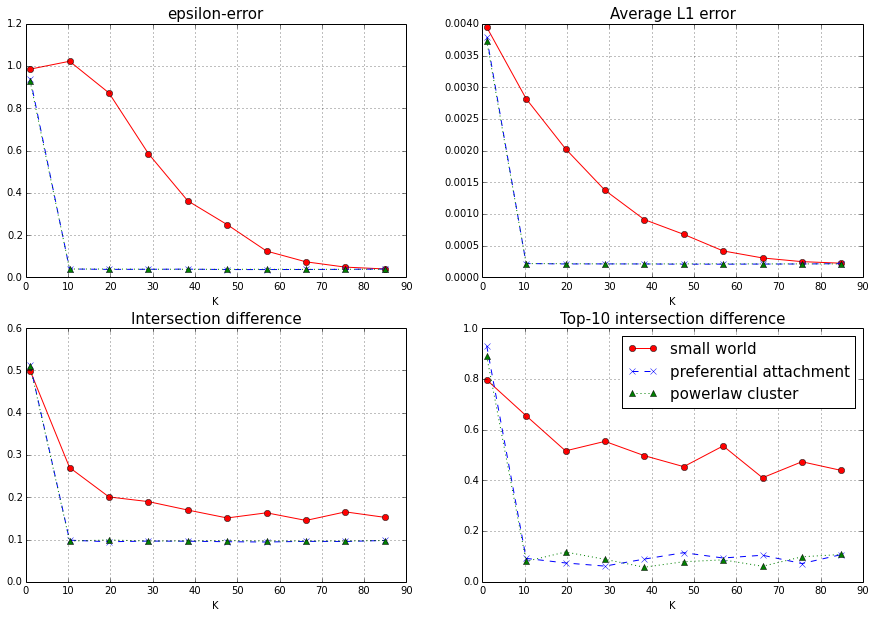}
\caption{Different measures of error for random graphs on $500$ when approximating heat
kernel pagerank with varying random walk lengths.}
\label{fig:random_rank_500}
\end{figure}

In both sets of plots and for every measure of error, we see that in the
preferential attachment and powerlaw graphs the error is minimized after
limiting random walks to only length $K=10$, regardless of the size.  We observe
a shallower decline in $\epsilon$-error, average $L_1$ error, and intsersection
differance for the small world graphs.  In particular, we note that the
intersection difference drops significantly after $10$ random walk steps for all
random graphs on both $100$ and $500$ vertices.  For $\epsilon = 0.1$, $K =
4\cdot \kparam \approx 11.04$ is enough to approximate the rankings for the
purpose of local clustering.

\subsection{Real graphs}
\label{sec:realranking}

\subsubsection{Network data}
For the experiments in this section, and later in Section~\ref{sec:expresults},
we use the following graphs compiled from real data.  The network data is
summarized in Table~\ref{table:realgraphs}.

\begin{enumerate}
\item \textbf{(dolphins)} A dolphin social network consisting of two
families~\cite{dolphins}.  The seed vertex is chosen to be a prominent member of
one of the families.\label{pt:dolphins}
\item \textbf{(polbooks)} A network of books about US politics published around
the time of the 2004 Presidential election and sold on Amazon~\cite{polbooks}.
Edges represent frequent copurchases.\label{pt:polbooks}
\item \textbf{(power)} The topology of the US Western States Power Grid
~\cite{powergrid}.\label{pt:powergrid}
\item \textbf{(facebook)} A combined collection of Facebook ego-networks,
including the ego vertices themselves~\cite{facebook}.\label{pt:facebook}
\item \textbf{(enron)} An Enron email communication network~\cite{enron},
in which vertices represent email addresses and an edge $(i,j)$ exists if an
address $i$ sent at least one email to address $j$.\label{pt:enron}
\end{enumerate}

\begin{table}
\centering
\begin{tabular}{|p{2cm}|l|c|c|c|}
\hline
\multicolumn{1}{|c|}{Network} & \multicolumn{1}{|c|}{Source} & ~~$|V|$~~ & ~~$|E|$~~ & Average degree\\
\hline\hline
dolphins & Dolphins social network~\cite{dolphins} & $62$ & $159$ & $5$\\\hline
polbooks & Copurchases of political books~\cite{polbooks} & $105$ & $441$ &
$8.8$\\\hline
power & Power grid topology~\cite{powergrid} & $4941$ & $6594$ & $2.7$\\\hline
facebook & Facebook ego-networks~\cite{facebook} & $4039$ & $88234$ &
$43.7$\\\hline
enron & Enron communication network~\cite{enron} & $36692$ & $183831$ &
$10$\\
\hline
\end{tabular}
\caption{Graphs compiled from real data.}
\label{table:realgraphs}
\end{table}

The network data for graphs~\ref{pt:dolphins},~\ref{pt:polbooks}, and
\ref{pt:powergrid} were taken from Mark Newman's network data
collection~\cite{newmandata}.  The network data for graphs~\ref{pt:facebook} and
\ref{pt:enron} are from the SNAP Large Network Dataset
Collection~\cite{snapdata}.

\subsubsection{Procedure} In this series of experiments, the seed vertex $u$ was
chosen to be a known member of a cluster.  As before, $t$ was chosen according
to $t=\tparamcluster$ with the values in Table~\ref{table:realrankingparams}.
For each graph and for each value of $K$ we compare an exact heat kernel
pagerank vector $\rho_{t,u}$ with an $\epsilon$-approximate heat kernel pagerank
vector $\hat{\rho}_{t,u}$.  Specifically, we consider the average $L_1$ distance
(\ref{eq:measure-l1err}) and the intersection difference
(\ref{eq:measure-isim}).  We again choose $K$ to range from $1$ to $t$.

\begin{table}
\centering
\begin{tabular}{|c|c|c|c|c|}
\hline
~~$\epsilon$~~ & Target        & Target       & Target         & $t$\\
               & Cheeger ratio & cluster size & cluster volume &\\
\hline
$0.1$ & $\phi = 0.05$ & $s = 100$ & $\varsigma = 1000$ & $95.6$\\
\hline
\end{tabular}
\caption{Parameters used to compute $t$ for vector computations.}
\label{table:realrankingparams} 
\end{table}

\subsubsection{Discussion}
Figure~\ref{fig:real_l1} plots the average $L_1$ error on the y-axis against
different values of $K$ on the x-axis for each of the dolphins, polbooks, and
power graphs.  Figure~\ref{fig:real_rank_isim} plots the intersection difference
on the y-axis against $K$ on the x-axis.

\begin{figure}
\centering
\includegraphics[width=\textwidth]{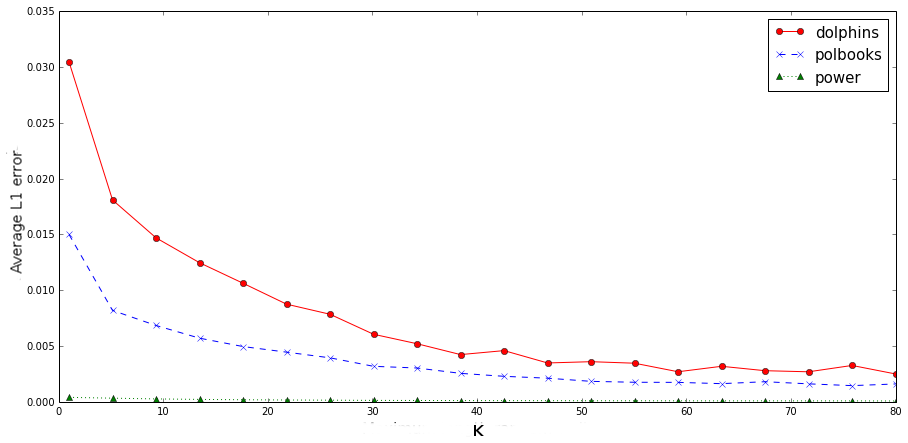}
\caption{Average error in each component for $\epsilon$-approximate heat kernel
pagerank vectors when allowing varying random walk lengths.}
\label{fig:real_l1}
\end{figure}

\begin{figure}
\centering
\includegraphics[width=\textwidth]{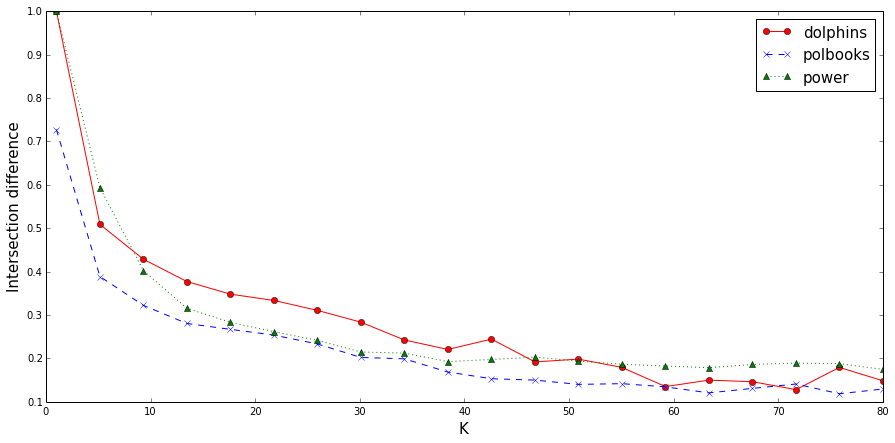}
\caption{Intersersection difference of the ranked lists of vertices computed by
exact and $\epsilon$-approximate heat kernel pagerank vectors when allowing
varying random walk lengths.}
\label{fig:real_rank_isim}
\end{figure}

First we discuss the average $L_1$ error.  The dolphins and the polbooks graphs
exihibit properties of both the small world graphs and the preferential
attachment graphs of the previous section (Figures~\ref{fig:random_rank_100}
and~\ref{fig:random_rank_500}).  Like the preferential attachment models, there
is a significant drop in average $L_1$ error after $K=5$, and like the small
world model the error continues to drop for larger values of $K$, approaching a
minimum error of $\approx 0.003$.  The average $L_1$ error in the power graph,
on the other hand,  is small for all values of $K$.  We remark that,
representing a power grid, the graph has very small average vertex degree, so
few random walk steps are enough to approximate the stationary distribution.

As for the intersection difference, we observe a smaller variance in values for
the three graphs.  Regardless of the size or structure of the graph, the
intersection difference drops sharply from $K=1$ to $K=5$.  For values larger
than $K = 10 < 4\cdot \kparam$, where $\epsilon = 0.1$, the intersection
difference decreases only marginally.

\paragraph{}The purpose of these experiments was to evaluate how error and differences of
ranking change in heat kernel pagerank approximation when varying $K$, the upper
bound on number of steps taken in random walks.  We found that setting an upper
bound for random walk lengths to $K = 10 < 4\cdot \kparam$ with $\epsilon = 0.1$
according to Theorem~\ref{thm:hkpraccuracy} yields approximations which satisfy
the prescribed error bounds.  This value is independent of the size of the graph
and $t$, and depends only on $\epsilon$.  Namely, we observed that choosing $K$
this way results in a significant decrease in both average $L_1$ error and
intersection difference as compared to smaller values of $K$, and only slight
decrease in average $L_1$ error and intersection difference for larger values of
$K$ as demonstrated in
Figures~\ref{fig:random_rank_100},~\ref{fig:random_rank_500},~\ref{fig:real_l1},
and~\ref{fig:real_rank_isim}.  Further, we tested graphs of various size, random
graphs generated from various models (see Section~\ref{sec:synthranking}), and
graphs from real data representing social networks, copurchasing networks, and
topological grids (see Section~\ref{sec:realranking}).  We found this choice of
$K$ was optimal for every graph regardless of size or structure.  That is, the
cutoff for random walk lengths does not depend on the size of the graph.

It is also worth mentioning that the most striking outlier among the subject
graphs is the small world graph, or expander graphs.  This is due to the fact
that the graph consists of a single cluster, which makes local cluster detection
ineffective.

\section{An Assessment of Cheeger Ratios Obtained with Local Clustering Algorithms}
\label{sec:expresults}
The goal of this section is to analyze the quality of local clusters computed
with a sweep over an approximate heat kernel pagerank vector (see
Section~\ref{sec:goodcuts} for details on sweeps).  We consider two objectives
for analysis.

The first objective is to validate the statement of Theorem~\ref{thm:localpart}.
To do this, we show that the Cheeger ratios of local clusters computed with
sweeps over approximate heat kernel pagerank vectors are within the
approximation guarantees of Theorem~\ref{thm:localpart}.  We use a slightly
modified version of \partitionalg~to compute these clusters.  We call this
modified algorithm $\epsilon$HKPR, and it is described in the list below.  

The second objective is to compare clusters computed with sweeps over different
vectors.  Namely, for a given graph and seed vertex, we compare the local
clusters computed using the following sweep algorithms:

\begin{enumerate}
\item \textbf{($\epsilon$HKPR)} A sweep over an $\epsilon$-approximate heat
kernel pagerank vector is performed.  The segment $S$ with volume $\vol(S) \leq
\vol(G)/2$ of minimal Cheeger ratio is output.  This is the
\partitionalg~algorithm with the following modification: we allow segments of
volume up to $\vol(G)/2$ rather than limiting the search to segments of volume
$< 2\varsigma$, twice the target volume.\label{pt:epshkpr}
\item \textbf{(HKPR)} A sweep over an exact heat kernel pagerank vector is
performed.  The segment $S$ with volume $\vol(S) \leq \vol(G)/2$ of minimal Cheeger
ratio is output.  This algorithm was outlined, but not stated explicitely,
in~\cite{chung:partitionhkpr:im09}.
\item \textbf{(PR)} A sweep over a Personalized PageRank vector
(\ref{eq:pagerank}) is performed.  The segment $S$ with volume $\vol(S) \leq
\vol(G)/2$ of minimal Cheeger ratio is output.  This is an adaptation of the
algorithm \texttt{PageRank-Nibble}\cite{acl:prgraphpartition:focs06} with the
following modifications: (i) rather than performing a sweep over an approximate
PageRank vector, perform a sweep over an exact PageRank vector, and (ii) allow
segments only as large as $\vol(G)/2$.
\end{enumerate}

We summarize the algorithms and parameters below in
Table~\ref{table:clusteralgs}.

\begin{table}
\centering
\begin{tabular}{|l|c|l|l|}
\hline
\multicolumn{1}{|c|}{Algorithm} & Sweep vector &
\multicolumn{1}{|c|}{Algorithm parameters} & \multicolumn{1}{|c|}{Sweep vector
parameters}\\
\hline\hline
$\epsilon$HKPR & $\hat{\rho}_{t,u}$ & $\phi$, target Cheeger ratio & $t = \tparamcluster$\\
& & $s$, target cluster size & $u$, seed vertex\\
& & $\varsigma$, target cluster volume & $\epsilon$, approximation parameter\\\hline
HKPR & $\rho_{t,u}$ & $\phi$, target Cheeger ratio & $t = 2\phi^{-1}\log s$\\
& & $s$, target cluster size & $u$, seed vertex\\\hline
PR & $\pr_{\alpha,u}$ & $\phi$, target Cheeger ratio & $\alpha = \phi^2/255\ln(100\sqrt{m})$\\
& & & $u$, seed vertex\\\hline
\end{tabular}
\caption{Algorithms used for comparing local clusters.}
\label{table:clusteralgs}
\end{table}

Each trial will resemble Procedure~\ref{alg:compareclusters}, as stated below.

\begin{algorithm}[H]
\floatname{algorithm}{Procedure}
\caption{Compare Clusters}
\label{alg:compareclusters}
\begin{algorithmic}
\State Let $G$ be a graph and $u$ a seed vertex
\State Choose parameters $\phi$, $s$, $\varsigma$, $\epsilon$
\State Let $S_A$ be a local cluster computed using the algorithm $\epsilon$HKPR
\State Let $S_B$ be a local cluster computed using the algorithm HKPR
\State Let $S_C$ be a local cluster computed using the algorithm PR
\State Compare $S_A, S_B, S_C$.
\end{algorithmic}
\end{algorithm}

The following sections describe the experiments in more detail.

\subsection{Synthetic graphs}
In this section, we use graphs generated with three random graph models:
Watts-Strogatz small world, Barab\'{a}si-Albert preferential attachment, and
Holme and Kim's powerlaw cluster as described in Section~\ref{sec:synthranking}.

\subsubsection{Procedure} We perform twenty-five trials of
Procedure~\ref{alg:compareclusters} and take the averages of Cheeger ratios and
cluster volumes computed.  Specifically, we fix a model and algorithm
parameters.  Then, generate a random graph according to the model and
parameters.  For each random graph, pick a random seed vertex with probability
proportional to degree.  Then, for each seed vertex compute local clusters $S_A,
S_B, S_C$ using the algorithms $\epsilon$HKPR, HKPR, and PR, respectively.  We
then use the average Cheeger ratio and cluster volume of the $S_A, S_B, S_C$ for
comparison.  In Table~\ref{table:synthclusterparams} we summarize the parameters
used for each random graph model.

\begin{table}
\centering
\begin{tabular}{|p{2.6cm}|c|c|c|c|c|c|c|}
\hline
\multicolumn{1}{|c|}{Model}& ~~$|V|~~$ & ~$d$~ & ~~$p$~~ & ~~$\epsilon$~~ & Target        & Target       & Target\\
       &           &       &         &                & Cheeger ratio & cluster size & cluster volume\\
\hline\hline
small world  & $100$   & $5$ & $0.1$ & $0.1$ & $0.1$ & $20$  & $100$\\
             & $500$   & $5$ & $0.1$ & $0.1$ & $0.1$ & $100$ & $500$\\
             & $800$   & $5$ & $0.1$ & $0.1$ & $0.1$ & $100$ & $500$\\
             & $1000$  & $5$ & $0.1$ & $0.1$ & $0.1$ & $100$ & $500$\\\hline
preferential & $100$   & $5$ &   -   & $0.1$ & $0.1$ & $20$  & $100$\\
attachment   & $500$   & $5$ &   -   & $0.1$ & $0.1$ & $100$ & $500$\\
             & $800$   & $5$ &   -   & $0.1$ & $0.1$ & $100$ & $500$\\\hline
powerlaw     & $100$   & $5$ & $0.1$ & $0.1$ & $0.1$ & $20$  & $100$\\
cluster      & $500$   & $5$ & $0.1$ & $0.1$ & $0.1$ & $100$ & $500$\\
             & $800$   & $5$ & $0.1$ & $0.1$ & $0.1$ & $100$ & $500$\\
\hline
\end{tabular}
\caption{Algorithm parameters used to compare local clusters.}
\label{table:synthclusterparams} 
\end{table}

\subsubsection{Discussion} We address the first analytic objective listed in the
introduction of this section by discussing the clusters output by
$\epsilon$HKPR.  Namely, we compare the clusters computed with $\epsilon$HKPR to
the guarantees of Theorem~\ref{thm:localpart}.  The results for each graph are
given in Table~\ref{table:synthhkprclusterresults}.  The first three columns
indicate the random graph model and algorithm parameters used for each instance.
The last two columns demonstrate how the (average) Cheeger ratio of clusters
computed by $\epsilon$HKPR compare to the approximation guarantee of
Theorem~\ref{thm:localpart}.  Namely, Theorem~\ref{thm:localpart} states that
the cluster output will have Cheeger ratio $\leq \sqrt{8\phi}$ with high
probability.  In every case the Cheeger ratio is well within the approximation
bounds.

\begin{table}
\centering
\textbf{Synthetic graphs}\\
\begin{tabular}{|p{2.6cm}|c|c|c|c|}
\hline
\multicolumn{1}{|c|}{Model} & ~~$|V|$~~  & $\phi$, Target & Cheeger ratio output by &
~~~~~$\sqrt{8\phi}$~~~~~\\
             &            & Cheeger ratio & $\epsilon$HKPR          & \\
\hline\hline
small world  & $100$   & $0.1$ & $0.173557$ & $0.894427$\\
             & $500$   & $0.1$ & $0.47316$ & $0.894427$\\
             & $800$   & $0.1$ & $0.510597$ & $0.894427$\\
             & $1000$  & $0.1$ & $0.519399$ & $0.894427$\\\hline
preferential & $100$   & $0.1$ & $0.523929$ & $0.894427$\\
attachment   & $500$   & $0.1$ & $0.503542$ & $0.894427$\\
             & $800$   & $0.1$ & $0.491046$ & $0.894427$\\\hline
powerlaw     & $100$   & $0.1$ & $0.517521$ & $0.894427$\\
cluster      & $500$   & $0.1$ & $0.500312$ & $0.894427$\\
             & $800$   & $0.1$ & $0.494145$ & $0.894427$\\
\hline 
\end{tabular}
\caption{Cheeger ratios of cluster output by $\epsilon$HKPR.}
\label{table:synthhkprclusterresults}
\end{table}

The second objective is to compare clusters computed with the three different
local clustering algorithms $\epsilon$HKPR, HKPR, and PR.
Table~\ref{table:cheegersynth} is a collection of cluster statistics for the
trials.  For each graph instance we list the average Cheeger ratio and cluster
volume of local clusters computed using the PR, HKPR, and $\epsilon$HKPR
algorithms, respectively.  

\begin{table}
\centering
\textbf{Synthetic graphs}\\
\begin{tabular}{|p{2cm}|c|c|c|c|}
\hline
\multicolumn{1}{|c|}{Model} & $|V|$ & PR & HKPR & $\epsilon$HKPR\\
\hline\hline
small world  & 100  & $0.235159$ & $0.087723$ & $0.173557$\\
             & & (volume = $52.52$) & (volume = $171.28$) & (volume = $142$)\\\cline{2-5}
             & 500  & $0.244261$ & $0.062263$  & $0.47316$\\
             & & (volume = $190.16$) & (volume = $943.68$) & (volume = $206.64$)\\\cline{2-5}
             & 800  & $0.246564$ & $0.064599$  & $0.510597$\\
             & & (volume = $162.68$) & (volume = $1413.6$) & (volume = $209.6$)\\\cline{2-5}
             & 1000 & $0.245612$ & $0.064716$ & $0.519399$\\
             & & (volume = $584.56$) & (volume = $1907.4$) & (volume = $225.2$)\\\hline
preferential & 100  & $0.430071$ & $0.512819$ & $0.523929$\\
attachment   & & (volume = $471.2$) & (volume = $467.16$) & (volume = $468.16$)\\\cline{2-5}
             & 500  & $0.508305$ & $0.51018$ & $0.503542$\\
             & & (volume = $2461.96$) & (volume = $2459.4$) & (volume = $2463.28$)\\\cline{2-5}
             & 800  & $0.491046$ & $0.496369$ & $0.491046$\\
             & & (volume = $3964.17$) & (volume = $3971.17$) & (volume =
$3951.83$)\\\hline
powerlaw     & 100  & $0.426828$ & $0.505277$ & $0.517521$\\
cluster      & & (volume = $463.4$) & (volume = $465.44$) & (volume = $464.88$)\\\cline{2-5}
             & 500  & $0.487341$ & $0.507328$ & $0.500312$\\
             & & (volume = $2447.12$) & (volume = $2460.44$) & (volume =
$2446.28$)\\\cline{2-5}
             & 800  & $0.522281$ & $0.513365$ & $0.494145$\\
             & & (volume = $3947$) & (volume = $3966$) & (volume = $3947$)\\
\hline 
\end{tabular}
\caption{Cheeger ratios of clusters output by different local clustering algorithms on
synthetic data.}
\label{table:cheegersynth}
\end{table}

We remark that for each graph there is little variation in Cheeger ratio and
volume of clusters computed by the three different algorithms.  We also note
that there is no obvious trend as graphs get larger.  The small world graphs
demonstrate the greatest variation in cluster quality.  However, as mentioned in
Section~\ref{sec:rankings}, expander graphs, such as small world graphs, consist
of one large cluster.

It is worth noting that in some trials the output volume is significantly
greater than twice the target volume.  While this may seem like a contradiction,
it is a consequence of our implementation.  During a sweep one may choose to
output a cluster of minimal Cheeger ratio, or one that satisfies volume
constraints, or both.  We are interested in comparing Cheeger ratios and so
allow the sweep to continue checking clusters that are well beyond twice the
target volume.

\subsection{Real graphs}
For these trials we use graphs generated from real data summarized in
Section~\ref{sec:realranking}.

\subsubsection{Procedure} We compare clusters computed by each of the three
algorithms as outlined in Procedure~\ref{alg:compareclusters}.  In these trials
we fix the seed vertex to be a member of a cluster with good Cheeger ratio.
Using this seed vertex, we compare the clusters computed using the
$\epsilon$HKPR, HKPR, PR algorithms.

For each trial we use the parameters listed in Table~\ref{table:clusterparams}.
We note that in each case the target cluster volume is computed to be roughly
the target cluster size times the average vertex degree, and here we use
$\epsilon=0.1$.  
\begin{table}
\centering
\begin{tabular}{|p{2cm}|c|c|c|c|c|c|c|}
\hline
\multicolumn{1}{|c|}{Network} & $|V|$ & $|E|$ & Average & ~~$\epsilon$~~ & Target & Target & Target\\
        &       &       & degree  &                & Cheeger ratio & cluster size & cluster volume\\
\hline
dolphins  & $62$    & $159$    & $5$    & $0.1$ & $0.08$ & $20$  & $100$  \\
polbooks  & $105$   & $441$    & $8.8$  & $0.1$ & $0.05$ & $30$  & $270$  \\
power     & $4941$  & $6594$   & $2.7$  & $0.1$ & $0.05$ & $200$ & $600$  \\
facebook  & $4039$  & $88234$  & $43.7$ & $0.1$ & $0.05$ & $200$ & $2800$ \\
enron     & $36692$ & $183831$ & $10$   & $0.1$ & $0.05$ & $100$ & $1000$ \\
\hline
\end{tabular}
\caption{Graph and algorithm parameters used to compare local clusters.}
\label{table:clusterparams} 
\end{table}

\subsubsection{Discussion}
Table~\ref{table:realhkprclusterresults} lists ratios output by $\epsilon$HKPR compared with the approximation guarantees of
Theorem~\ref{thm:localpart}.  In each case, the Cheeger ratios are well within
the approximation bounds of Theorem~\ref{thm:localpart}.

\begin{table}
\centering
\textbf{Real graphs}\\
\begin{tabular}{|p{2cm}|c|c|c|}
\hline
\multicolumn{1}{|c|}{Network} & $\phi$, Target & Cheeger ratio output by &
~~~~~$\sqrt{8\phi}$~~~~~\\
        & Cheeger ratio & $\epsilon$HKPR & \\
\hline\hline
dolphins & $0.08$ & $0.083333$ & $0.8$ \\
polbooks & $0.05$ & $0.052133$ & $0.632456$\\
power    & $0.05$ & $0.346667$ & $0.632456$ \\
facebook & $0.05$ & $0.056939$ & $0.632456$ \\
enron    & $0.05$ & $0.036602$ & $0.632456$\\
\hline 
\end{tabular}
\caption{Cheeger ratios of cluster output \partitionalg.}
\label{table:realhkprclusterresults}
\end{table}

The complete numerical data obtained from the set of the trials are given in
Table~\ref{table:cheegerreal}.  For each graph we list the Cheeger ratio,
cluster volume, and additionally the cluster size of local clusters computed
using each of the algorithms PR, HKPR, and $\epsilon$HKPR, respectively.

\begin{table}
\centering
\textbf{Real graphs}\\
\begin{tabular}{|p{2cm}|c|c|c|}
\hline
\multicolumn{1}{|c|}{Network} & PR & HKPR & $\epsilon$HKPR\\
\hline\hline
dolphins & $0.226415$ & $0.163636$ & $0.083333$ \\
         & (volume = $106$) & (volume = $110$) & (volume = $96$)\\
         & (size = $23$) & (size = $24$) & (size = $20$)\\\hline
polbooks & $0.079518$ & $0.245657$ & $0.052133$\\
         & (volume = $415$) & (volume = $403$) & (volume = $422$)\\
         & (size = $48$) & (size = $49$) & (size = $50$)\\\hline
power    & $0.375$    & $0.002764$ & $0.346668$ \\
         & (volume = $16$) & (volume = $4342$) & (volume = $300$)\\
         & (size = $6$) & (size = $1564$) & (size = $85$)\\\hline
facebook & $0.418993$ & $0.001277$ & $0.056939$ \\
         & (volume = $88140$) & (volume = $67326$) & (volume = $35266$)\\
         & (size = $3063$) & (size = $1094$) & (size = $258$)\\\hline
enron    & $0.48797$ & - & $0.036602$\\
         & (volume = $183612$) & - &  (volume = $3579$)\\
\hline 
\end{tabular}
\caption{Cheeger ratios of cluster output by different local clustering
algorithms.}
\label{table:cheegerreal}
\end{table}

For each graph, the local cluster computed using $\epsilon$HKPR has smaller
Cheeger ratio than the local cluster computed using PR.  For the power graph, we
observe that the cluster of minimal Cheeger ratio was computed using the HKPR
algorithm, but it is nearly a third the size of the entire network.  The
algorithms $\epsilon$HKPR and PR, on the other hand, each return smaller
clusters.  We remark that for real graphs, the clusters computed using sweeps
over different vectors have more variation than for random graphs.

At this point we remark about our choice of parameters for the trials.  At this
point the sensitivity of the algorithm to the choice of $\epsilon, \phi, s,$ and
$\varsigma$ has not been fully explored.  In particular, it is worth studying
the effect of $\phi$ on the output cluster in future work.

To conclude, we include visualizations of clusters computed in the facebook
ego-network to illustrate the differences in local cluster detection.
Figure~\ref{fig:fb_ehkpr} colors the vertices in a local cluster computed using
the $\epsilon$HKPR algorithm, as described in Table~\ref{table:cheegerreal}.
Figure~\ref{fig:fb_pr} colors the vertices in a local cluster compted using the
PR algorithm.

\begin{figure}
\centering
\includegraphics[width=\linewidth]{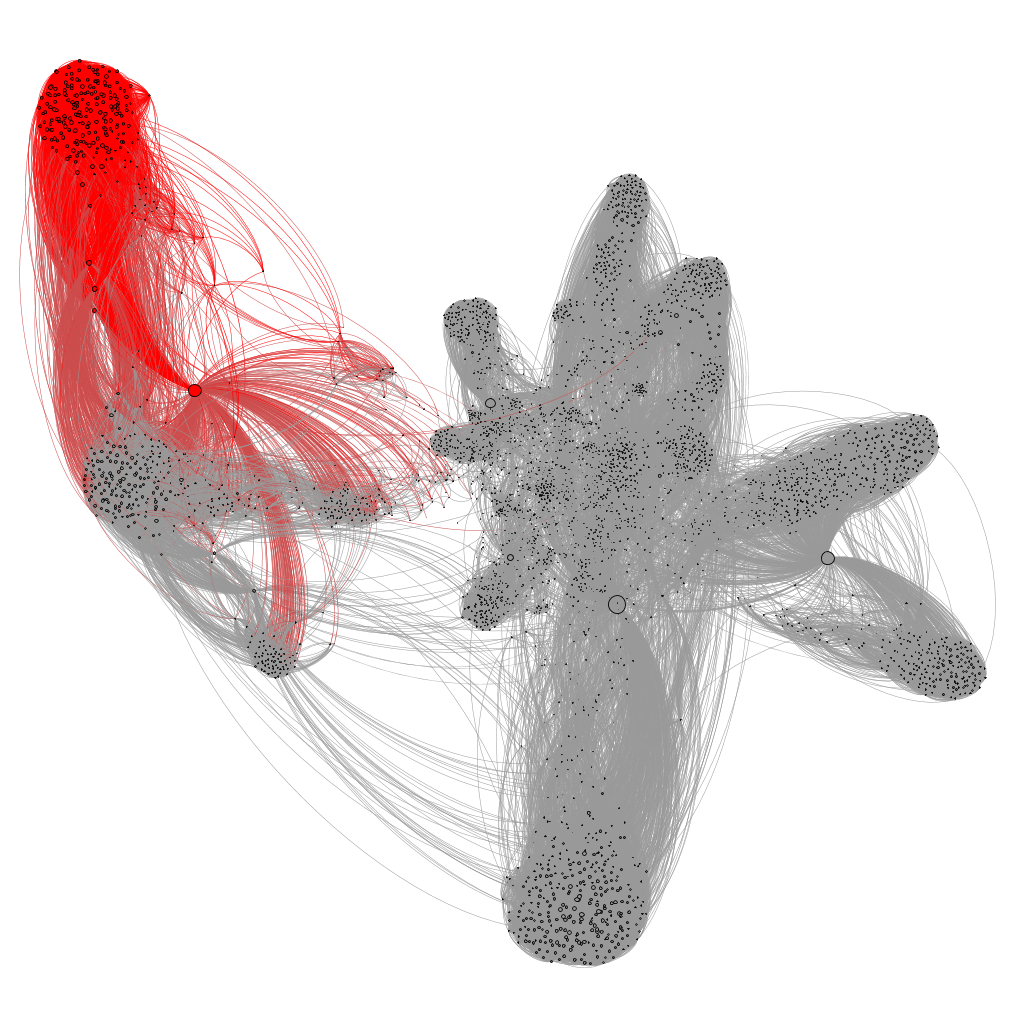}
\caption{Local cluster in facebook ego network computed using the $\epsilon$HKPR
algorithm.}
\label{fig:fb_ehkpr}
\end{figure}

\begin{figure}
\centering
\includegraphics[width=\linewidth]{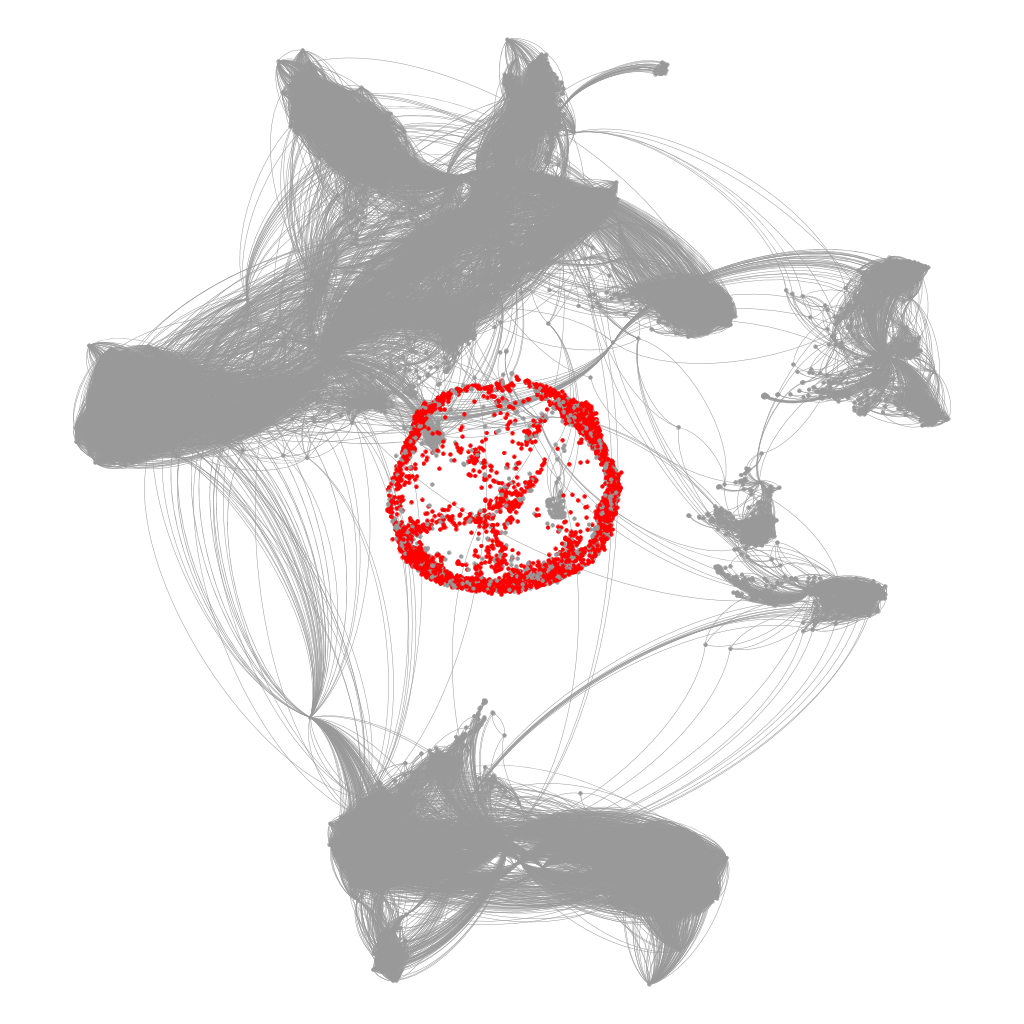}
\caption{Local cluster in facebook ego network computed using the PR
algorithm.}
\label{fig:fb_pr}
\end{figure}

The numerical data of the last two sections validate the effectiveness and
efficiency of local cluster detection using sweeps over $\epsilon$-approximate
heat kernel pagerank.  The experiments of Section~\ref{sec:rankings} demonstrate
that sampling a number of random walks of at most $K$ steps yield a ranking of
vertices within the error bounds of Theorem~\ref{thm:hkpraccuracy}.  This ranking
in turn is used to compute a local cluster.  What is more, this value $K$ does
not depend on parameters other than $\epsilon$.  Specifically, it does not
depend on the size of the graph or the desired cluster volume, size, or Cheeger
ratio.  Finally, the data of Section~\ref{sec:expresults} validate the
statements of Theorem~\ref{thm:localpart}.  That is, perfoming a sweep over an
approximate heat kernel pagerank vector detects clusters of Cheeger ratio at
most $\sqrt{8\phi}$ for a desired Cheeger ratio $\phi$.  The total cost of
computing this cluster is $\hkprcomplexity$, sublinear in the size of the graph.

%
\bibliographystyle{amsplain}
\bibliography{ejc_hkpr}  
%
%

\end{document}